\RequirePackage{fix-cm}
\documentclass[aA4paper,referee,envcountsect]{svjour3}
\smartqed  
\usepackage[dvips]{graphicx}
\usepackage{cite}
\usepackage{setspace}
\usepackage[titletoc,toc,title]{appendix}
\usepackage{geometry}
\usepackage{latexsym,amsfonts,amsmath,epsfig,amssymb,mathrsfs,bbding,color,subfigure}
\usepackage[misc]{ifsym}
\newcommand\blfootnote[1]{%
    \begingroup%
\let\thefootnote\relax\footnotetext{\hspace{-15pt}#1}%
\endgroup}%
\journalname{Journal of Optimization Theory and Applications}
\begin{document}
\doublespacing
\title{Robust Portfolio Optimization with Multi-Factor Stochastic Volatility\blfootnote{Communicated by Kok Lay Teo
\\\hrule height 0.4pt width 3.8cm}
}

\author{
Ben-Zhang Yang \and  Xiaoping Lu \and Guiyuan Ma \and Song-Ping Zhu}
\institute{
Ben-Zhang Yang \at Department of Mathematics, Sichuan University, Chengdu,  China\\
\email{yangbenzhang@126.com} \and
Xiaoping Lu (\Letter) \at
              School of Mathematics and Applied Statistics, University of Wollongong, Australia\\
\email{xplu@uow.edu.au} \and
Guiyuan Ma \at Department of Statistics, The Chinese University of Hong Kong, Hong Kong,  China\\
\email{maguiyuan@foxmail.com} \and
Song-Ping Zhu  \at
              School of Mathematics and Applied Statistics, University of Wollongong,  Australia\\
\email{spz@uow.edu.au}
}
\date{Received: date / Accepted: date}
\maketitle
\begin{abstract}
This paper studies a robust portfolio optimization problem under a multi-factor volatility model. We derive optimal strategies analytically under the worst-case scenario with or without derivative trading in complete and incomplete markets, and for assets with jump risk. We extend our study to the case with correlated volatility factors and propose an analytical approximation for the robust optimal strategy. To illustrate the effects of ambiguity, we compare our optimal robust strategy with the strategies that ignore the information of uncertainty, and provide the welfare analysis. We also discuss how derivative trading affects  the optimal strategies.
Finally, numerical experiments are provided to demonstrate the behavior of the optimal strategy and the utility loss.
\end{abstract}
\keywords{Robust portfolio selection \and Multi-factor volatility \and Jump risks \and  Non-affine stochastic volatility \and  Ambiguity effect}
\subclass{91B28 \and 60H30 \and 91C47 \and 91B70 }

\section{Introduction}
\label{intro}
\indent\par\setcounter{equation}{0}
\setcounter{lemma}{0}
\setcounter{theorem}{0}
\setcounter{remark}{0}
\setcounter{proposition}{0}
\setcounter{definition}{0}

During the past decades, various stochastic volatility models have been proposed to explain volatility smile, to address term structure effects, and to describe more complex financial markets (for example, \cite{Cui17, Cui19, Liu03, Ma2019, Nandi98, Pan02, Yang17, Yang18, Yue18,Ma19}). Stochastic volatility models also address term structure effects by modeling the mean reversion in variance dynamics. The existing literature includes, not only one-factor stochastic volatility model, such as \cite{Stein91,Heston93,Schobel98}, but also  multi-factor stochastic volatility model, such as \cite{Christoffersen09,Fonseca08,Li10}.

Optimal portfolio selection problems with multi-factor volatility have attracted a lot of attention in the recent literature{, with the view of multi-factor models potentially capturing market volatility better than classical single-factor models}. Escobar et al. \cite{Escobar17} considered an optimal investment problem under a multi-factor stochastic volatility. Assuming that the eigenvalues of the covariance matrix of asset returns  follow independent square-root stochastic processes, they derived the optimal investment strategies in closed form. In practice, investors quite often face model uncertainty about the probability distribution of the dynamic process  \cite{Faidi19,Sun19,Yuen12}. As a result, an investor may consider a robust alternative model for the stock and its volatilities to avoid miss-specification, when making investment decisions \cite{ait2019robust,ismail2019robust}. The investor should also identify the worst-case measure and specify the optimal portfolio under the worst-case scenario.  Bergen et al. \cite{Bergen18} studied a robust multivariate portfolio selection problem with stochastic covariance in the presence of ambiguity and provided the optimal multivariate intertemporal portfolio.


{ In this paper, we provide a study of the robust portfolio optimization problem with multi-factor stochastic volatility. It is worth pointing out that our study is fundamentally different from the works in\cite{Bergen18,Escobar17}, where there are multiple assets, each with an independent  single-factor stochastic volatility.
We consider only one risky asset, but with a multi-factor volatility structure whose components may be correlated. In addition, we also consider the robust portfolio selection problem in the presence of jump risk, which,  to the best of our knowledge, has not been studied in the literature of robust optimization under multi-factor volatility.}

In practice, the uncertainty of market price for volatility risk affects investment decisions. Liu \& Pan \cite{Liu03} and Larsen \& Munk \cite{Larsen12} explored the effect of such uncertainty { with only one single-factor volatility. Bergen et al. \cite{Bergen18} considered the ambiguity effect of multiple assets, each with one-factor volatility.  Our analysis incorporates the multi-factor model introduced by Christoffersen et al. \cite{Christoffersen09} with} the assumption that there exists some ambiguity in an investor's mind in terms of the asset dynamics and its volatilities.

A highlight of this paper is our study of the impact of the correlation between volatility factors in the context of robust portfolio optimization. The correlation is considered { being able} to capture the evolution of correlation between asset returns or multivariate volatilities \cite{Burn98,Grzelak11,Grzelak12}. Unlike the independent volatilities case, the non-affine structure of volatility with correlated factors rules out the possibility of finding a closed-form solution. However, the inclusion of correlated volatility processes still adds value to the current literature of robust portfolio optimization problems.

The rest of the paper is structured as follows. Our model formulation is introduced in Section 2. Section 3 presents { the analytical solutions of} the optimal and sub-optimal investment strategies for the worst-case measure in complete and incomplete markets, and for the case of asset price with jump risks.  Section 4 provides { an analytical approximation of the robust optimal strategy for the case with correlated volatility factors}. Numerical examples are provided in Section 5, and conclusions  are given in Section 6. 

\section{Basic Model}

We propose a portfolio optimization problem with one risky asset under multi-factor stochastic volatility structure by extending the one-factor volatility model in \cite{Liu03}. Under the assumption that investors have access to both stock and derivatives markets, the money market account follows
\begin{equation}\label{bond}
dM(t)=rM(t)dt,
\end{equation}
where  $r$ is a constant risk-free interest rate. In this paper, the stock price and its volatilities follow the multi-factor volatility model of \cite{Christoffersen09}. Without loss of generality, a two-factor stochastic volatility model  is presented in this paper. Specifically, the stock price satisfies:
\begin{equation}\label{S1}
dS(t)=[r+\sum_{j=1}^{2}\lambda_jV_j(t)] S(t)dt+\sum_{j=1}^{2}\sqrt{V_j(t)}S(t)dW_{j}(t),\  j = 1, 2,
\end{equation}
{ where $W_j(t)$ are independent Brownian motions, and the variance $V_j(t)$ are assumed to follow}
\begin{equation}\label{V0}
dV_j(t)=\kappa_j(\theta_j-V_j(t))dt+\sigma_j\sqrt{V_j(t)}\left(\rho_j dW_{j}(t)+\sqrt{1-\rho_j^2} dZ_j(t)\right),\ { j =1, 2,}
\end{equation}
where $Z_j(t)$ are another two independent Brownian motions; $\rho_j$ are the correlation parameters;
$\kappa_j$, $\theta_j$ and $\sigma_j$ are the mean-reverting speed, the long-term mean and  the volatility of volatility in $V_j(t)$, respectively. The Feller condition is assumed to be satisfied, i.e. $2\kappa_j\theta_j\geq \sigma_j^2$ holds (see \cite{Heston93}).

{ Let $\lambda_j \sqrt{V_j(t)}$ and $\mu_j \sqrt{V_j(t)}$  be the market  prices of the risk associated with $W_j$ and $Z_j$, respectively,  where $\lambda_j$ and $\mu_j$ are constant risk premium parameters. In this section and Section 3, the volatility components are assumed to be independent. As a result, the variance is the sum of the two uncorrelated factors, each of which may be individually correlated with stock returns \cite{Christoffersen09}.
Due to the variance structure,  an investor needs to trade  three  options to hedge the market risks associated with each component of the volatility and the stock. Let the option price be $O^{(i)}(t)=g^{(i)}(S,V_1,V_2,t)$ for some twice continuously differentiable function $g^{(i)}$,  $i=1,2,3$. Using It\^o's lemma, we obtain the following option price dynamics}
\begin{equation}\label{Oi}
dO^{(i)}(t)=rO^{(i)}(t)dt+\sum_{j=1}^{2}\left[(g^{(i)}_SS+\sigma_j\rho_jg^{(i)}_{V_j})(\lambda_jV_jdt+\sqrt{V_j}dW_j+\sigma_{j}\sqrt{1-\rho_j^2}g^{(i)}_{V_j}(\mu_jV_jdt+\sqrt{V_j}dZ_j)\right],
\end{equation}
where $g^{(i)}_S$ and $g^{(i)}_{V_j}$ denote the partial derivatives of $g^{(i)}$ with respect to $S$ and $V_j$, respectively. { Equations \eqref{S1}-\eqref{Oi} are referred to as} our reference model.

In reality, an investor { facing uncertainty about the probability distribution for the reference model would often consider} a set of possible alternative models when making investment decisions. We assume that the investor is uncertain about the distribution of noises $W_j$ and $Z_j$ in the asset price and its volatility processes, and that $\mathcal{F}_t$ is the filtration generated by Brownian motions $W_j$ and $Z_j$.  Let perturbation process $\mathbf{e}=\left(e^S_1(t),e^S_2(t),e^V_1(t),e^V_2(t)\right)$ be a $\mathbb{R}^4$-valued $\mathcal{F}_t$-progressively measurable process and $\mathcal{E}[0,T]$ be the space of all $\mathcal{F}_t$-measurable processes such that $\mathbf{e}$ is a well-defined Radon-Nikodym derivative process:
\begin{eqnarray}\label{RD}
  \mathcal{Z}_t^\mathbf{e}=\mathrm{E}\left[\left.\frac{d\mathbb{P}^\mathbf{e}}{d\mathbb{P}}\right|\mathcal{F}_t\right]
=\exp\bigg(-\int_{0}^{t}\sum_{j=1}^{2}\big[\frac{1}{2}\left((e^S_j(\tau))^2+(e^V_j(\tau))^2\right)d\tau
+e^S_j(\tau)dW_j(\tau)+e^V_j(\tau)dZ_j(\tau)\big]\bigg).
\end{eqnarray}
According to Girsanov's theorem, the processes defined as
$$
\widehat{W}_j(t)=\int_{0}^{t}e^S_j(\tau)d\tau+W_j(t), \quad\widehat{Z}_j(t)=\int_{0}^{t}e^V_j(\tau)d\tau+Z_j(t),
$$
are Brownian motions under the probability measure $\mathbb{P}^e$. Due to the difficulty identifying the reference model from the available market data, for each perturbation process $\mathbf{e}$, the investor considers an alternative model, in which the stock price follows the process
\begin{equation}\label{SE}
dS(t)=\Big(r+\sum_{j=1}^{2}\left[\lambda_jV_j(t)-\sqrt{V_j(t)}e^S_j\right]\Big) S(t)dt+\sum_{j=1}^{2}\sqrt{V_j(t)}S(t)d\widehat{W}_{j}(t),
\end{equation}
and at the same time, its variance processes $V_j$, $j=1, 2$, are governed by
\begin{equation}\label{VE}
 dV_j(t)  =
\left[\kappa_j(\theta_j-V_j(t))-\rho_j\sigma_j\sqrt{V_j}e^S_j-\sqrt{1-\rho^2_j}\sigma_j\sqrt{V_j}e^V_j\right]dt+\sigma_j\sqrt{V_j(t)}\left(\rho_j d\widehat{W}_{j}(t)+\sqrt{1-\rho_j^2} d\widehat{Z}_j(t)\right),
\end{equation}
while the option prices $O^{(i)}, i = 1, 2, 3$, satisfy
\begin{equation}\label{OE}
\begin{aligned}
dO^{(i)}(t) = &  rO^{(i)}(t)dt+\sum_{j=1}^{2}\sigma_{j}\sqrt{1-\rho_j^2}g^{(i)}_{V_j}\left((\mu_jV_j-e^V_j\sqrt{V_j})dt+\sqrt{V_j}d\widehat{Z}_j(t)\right)\\
&+\sum_{j=1}^{2}\left(g^{(i)}_SS+\sigma_j\rho_jg^{(i)}_{V_j}\right)\left((\lambda_jV_j-e^S_j\sqrt{V_j})dt+\sqrt{V_j}d\widehat{W}_j(t)\right){.}
\end{aligned}
\end{equation}

Under the actual trading scenario, different investors  obtain { their} information about the probability distributions of the stock price and the volatility processes from different sources. { Consequently}, each investor chooses one certain measure, which is determined by { his or her} perturbation process $\mathbf{e}$. For each {chosen} perturbation process, the {investor  considers} the corresponding alternative model. { Since  the perturbation process can be chosen arbitrarily, the alternative model} allows for different levels of ambiguity about the stock and its volatility. Naturally the investor  attempts to find a robust decision based on the sources which are more reliable and useful.

Let $X_t$ be an investor's total wealth, $\pi^S$ be the fraction invested in stock, $\pi^{i}$ $(i=1,2,3)$ be the fraction invested in the $i$-th option, and the rest of the wealth in a money market account, so their investment strategy is $\Pi=(\pi^S,\pi^1,\pi^2,\pi^3)$.
Assume that the ambiguous investor wishes to derive an optimal strategy maximizing the expected utility of terminal wealth $X_{T}$, {  and that the investor's preferences are described by a CRRA (Constant  Relative Risk Aversion) utility function with parameter $\gamma>1$ as in \cite{Uppal03,Flor14,Escobar15}.  Let $\mathcal{U}[0,T]$ be the set of all admissible strategies $\Pi$ statisfying} the following conditions: (i) $\Pi$ is a $\mathcal{F}_t$-progressively measurable process; (ii) Under $\Pi$, the wealth process $X_t$ of the investor is non-negative for $t\in[0,T]$; (iii) The integrability conditions{, which are necessary for the expectation operator in \eqref{max1a} to be well-defined, } are satisfied. Denote a new process $\mathbf{Y}(s)=(X(s),V_1(s),V_2(s))$ and let $\mathbf{y}=(x,v_1,v_2)$ be the values of $\mathbf{Y}(s)$ at time $t$, then the expected utility {corresponding to} a trading strategy $\Pi \in\mathcal{U}[0,T] $ is given by
\begin{equation}\label{max1a}
w^\mathbf{e}(t,\mathbf{y};\Pi)=\frac{1}{1-\gamma}\mathrm{E}^{\mathbb{P}^\mathbf{e}}_{t,\mathbf{y}}\left[(X_T)^{1-\gamma}\right].
\end{equation}
Thus, the indirect utility function of the investor is defined as
\begin{equation}\label{max}
J(t,\mathbf{y})=\sup_{\Pi }\inf_{\mathbf{e} }\hskip -2pt
\Big(\hskip -2pt w^\mathbf{e}(t,y;\Pi)+\mathrm{E}^{\mathbb{P}^\mathbf{e}}_{t,\mathbf{y}}\hskip -2pt\Big[\hskip -2pt\int_{t}^{T}\hskip-5pt\sum_{j=1}^2\frac{(e^S_j(s))^2}{2\Psi^S_j(s,Y)}+\frac{(e^V_j(s))^2}{2\Psi^V_j(s,Y)}ds
\Big]\Big).
\end{equation}
Here the expectation is taken with respect to the distribution from the alternative model, and the integral term in (\ref{max}) is the penalty incurred by deviating from the reference model. As pointed out by Anderson et al. \cite{Anderson03}, the penalty term is an expected log-likelihood ratio on the basis of the relative entropy. The state-dependent scaling functions $\Psi^S_j$ and $\Psi^V_j$ in \eqref{max} are defined below for analytical tractability as in \cite{Flor14,Escobar15,Escobar17}
\begin{equation}
\Psi^S_j=\frac{\phi^S_j}{(1-\gamma)J(t,\mathbf{y})},
\quad
\Psi^V_j=\frac{\phi^V_j}{(1-\gamma)J(t,\mathbf{y})},
\end{equation}
where {positive} constants $\phi^S_j$ and $\phi^V_j$ are  ambiguity aversion parameters that describe the ambiguity aversion level about { the} stock price and its volatilities, respectively.
{ Functions} $\Psi^S_j$ and $\Psi^V_j$ represent the strength of the investor's preference for robustness,  with greater values reflecting less faith in the reference model.
\section{Robust Optimal Investment Strategies}
\subsection{{\textbf{Complete Market Case}}}
In the complete market, the investor's wealth process {satisfies }
\begin{equation}\label{wealth}
\begin{aligned}
\frac{dX(t)}{X(t)}&=\pi^S\frac{dS(t)}{S(t)}+\sum_{i=1}^{3}\pi^{i}\frac{dO^{i}(t)}{O^{i}(t)}+\left(1-\pi^S-\sum_{i=1}^3\pi^{i}\right)rdt\\
&=rdt+\sum_{j=1}^{2}\Big[\big(\beta^{S}_j(\lambda_jV_j-\sqrt{V_j}e^S_j)
+\beta^{V}_j(\mu_jV_j-\sqrt{V_j}e^V_j)\big)dt+\sqrt{V_j}\big(\beta^{S}_jd\widehat{W}_j(t)+\beta^{V}_jd\widehat{Z}_j(t)\big)\Big],
\end{aligned}
\end{equation}
where
\begin{equation}\label{eve}
  \beta^S_j=\pi^S+\sum_{i=1}^{3}\frac{g^{(i)}_S+\sigma_j\rho_jg^{(i)}_{V_j}}{O^{i}(t)}\pi^i \quad \text{and}\quad
\beta^V_j=\sum_{i=1}^{3}\frac{\sigma_j\sqrt{1-\rho_j^2}g^{(i)}_{V_j}}{O^{i}(t)}\pi^i
\end{equation}
represent the investor's wealth exposure to   risk factors $W_j$ and $Z_j$, respectively, and in matrix form
\begin{equation}\label{hhh}
\begin{bmatrix}\beta^S_1 \\ \beta^S_2 \\  \beta^V_1 \\ \beta^V_2 \end{bmatrix}
=
\begin{bmatrix}
1 & \frac{g^{(1)}_S+\sigma_1\rho_1g^{(1)}_{V_1}}{O^{1}(t)} & \frac{g^{(2)}_S+\sigma_1\rho_1g^{(2)}_{V_1}}{O^{2}(t)} & \frac{g^{(3)}_S+\sigma_1\rho_1g^{(3)}_{V_1}}{O^{3}(t)} \\
1 & \frac{g^{(1)}_S+\sigma_2\rho_2g^{(1)}_{V_2}}{O^{1}(t)} & \frac{g^{(2)}_S+\sigma_2\rho_2g^{(2)}_{V_2}}{O^{2}(t)} & \frac{g^{(3)}_S+\sigma_2\rho_2g^{(3)}_{V_2}}{O^{3}(t)} \\
0 & \frac{\sigma_1\sqrt{1-\rho_1^2}g^{(1)}_{V_1}}{O^1(t)} & \frac{\sigma_1\sqrt{1-\rho_1^2}g^{(2)}_{V_1}}{O^2(t)} & \frac{\sigma_1\sqrt{1-\rho_1^2}g^{(3)}_{V_1}}{O^3(t)}\\
0 & \frac{\sigma_2\sqrt{1-\rho_2^2}g^{(1)}_{V_2}}{O^1(t)} & \frac{\sigma_2\sqrt{1-\rho_2^2}g^{(2)}_{V_2}}{O^2(t)} & \frac{\sigma_2\sqrt{1-\rho_2^2}g^{(3)}_{V_2}}{O^3(t)}\\
\end{bmatrix}
\begin{bmatrix}\pi^S \\ \pi^1 \\  \pi^2 \\ \pi^3 \end{bmatrix}
=:A\begin{bmatrix}\pi^S \\ \pi^1 \\  \pi^2 \\ \pi^3 \end{bmatrix}.
\end{equation}
The matrix $A$ in \eqref{hhh} is assumed to { be} full rank to keep the completeness of the market with respect to the chosen derivative securities, the risky stock, and the money market account (see also \cite{Liu03, Escobar15}). Thus, any exposure can be achieved by taking an appropriate position in a complete market. To build our analysis independent of the derivative security chosen, we { examine} the exposures, instead of the portfolio weights considered in \cite{Liu03,Escobar15,Escobar17}.

{ The value function} $J(t,\mathbf{y})$ in \eqref{max} satisfies the robust Hamilton-Jacobi-Bellman (HJB) PDE:
\begin{equation}\label{HJB0}
\begin{aligned}
\sup_{\beta_j^S\hskip-2pt,\, \beta^V_j}&\inf_{e^S_j\hskip -2pt,\, e^V_j}
\Big\{
J_t+x\Big[r+\sum_{j=1}^{2}\big(\beta^{S}_j\lambda_jv_j-\beta^{S}_j\sqrt{v_j}e^S_j
+\beta^{V}_j\mu_jv_j-\beta^{V}_j\sqrt{v_j}e^V_j\big)\Big]J_x
+\frac{1}{2}x^2\sum_{j=1}^{2}\big[(\beta^{S}_j)^2\\
&+(\beta^{V}_j)^2\big]v_jJ_{xx}+\frac{1}{2}\sum_{j=1}^{2}\sigma_j^2v_jJ_{v_jv_j}
+\sum_{j=1}^{2}\Big[\kappa_j(\theta_j-v_j)-\rho_j\sigma_j\sqrt{v_j}e^S_j-\sqrt{1-\rho_j^2}\sigma_j\sqrt{v_j}e^V_j\Big]J_{v_j}\\
&+\sum_{j=1}^{2}\hskip -3pt\Big[\sigma_jv_jx\Big(\beta^S_j\rho_j+\beta^V_j\hskip -3pt\sqrt{1-\rho_j^2}\Big)J_{v_jx}
+\frac{(e^S_j)^2}{2\Psi^S_j}+\frac{(e^V_j)^2}{2\Psi^V_j}\Big]\hskip -2pt\Big\}=0.
\end{aligned}
\end{equation}
The solution to \eqref{HJB0} is provided in the following proposition.
\begin{proposition}\label{prop1}
In a complete market, the indirect utility  of an ambiguity and risk averse investor is
\begin{equation}
J(t,x,v_1,v_2)=\frac{x^{1-\gamma}}{1-\gamma}\exp\left[H_1(\tau)v_1+H_2(\tau)v_2+h(\tau)\right],
\end{equation}
\begin{equation}\label{Hh1}
\begin{aligned}
&H_j(\tau)=\frac{2c_j(1-e^{-d_j\tau})}{2d_j+(a_j+d_j)(e^{-d_j\tau}-1)},\  j = 1, 2, \\
&h(\tau)=(1-\gamma)r\tau-\sum_{j=1}^{2}\kappa_j\theta_j\left(\frac{a_j+d_j}{2b_j}\tau+\frac{1}{b_j}\ln\left(\frac{e^{-d_j\tau(a_j+d_j)}-a_j+d_j}{2d_j}\right)\right),
\end{aligned}
\end{equation}
where {$\tau=T-t$, constants $a_j, b_j, c_j$ and $d_j$ are given in Appendix \ref{A1}.}

The optimal exposures to the risk factors $W_j$ and $Z_j$ $(j=1,2)$ are
\begin{equation}\label{e17}
\beta^S_j=\frac{\lambda_j}{\gamma+\phi^S_j}+\frac{(1-\gamma-\phi^S_j)\sigma_j\rho_j}{(1-\gamma)(\gamma+\phi^S_j)}H_j(\tau),\quad \beta^V_j=\frac{\mu_j}{\gamma+\phi^V_j}+\frac{(1-\gamma-\phi^V_j)\sigma_j\sqrt{1-\rho_j^2}}{(1-\gamma)(\gamma+\phi^V_j)}H_j(\tau).
\end{equation}
The worst-case measure are given by
\begin{equation}\label{opte}
e^S_j=\left(\frac{\lambda_j}{\gamma+\phi^S_j}+\frac{\sigma_j\rho_j H_j(\tau)}{(1-\gamma)(\gamma+\phi^S_j)}\right)\phi^S_j\sqrt{v_j},\quad e^V_j=\left(\frac{\mu_j}{\gamma+\phi^V_j}+\frac{\sigma_j\sqrt{1-\rho_j^2} H_j(\tau)}{(1-\gamma)(\gamma+\phi^V_j)}\right)\phi^V_j\sqrt{v_j}.
\end{equation}
\end{proposition}
\begin{proof}
See Appendix \ref{A1}.
\end{proof}
\begin{remark}
It is well-known that there exists a unique equivalent risk-neutral measure $\mathbb{P}$ in the complete market. Because the optimization problem \eqref{max} under measure $\mathbb{P}^\mathbf{e}$ is equivalent to
\begin{equation*}
\max_{\Pi \in \mathcal{U}(t,T)} \frac{1}{1-\gamma}\mathrm{E}^{\mathbb{P}}_{t,\textbf{y}}\left[(X_T)^{1-\gamma}\right],
\end{equation*}
under measure $\mathbb{P}$, the complete market condition determines the uniqueness of the worst-case measure and ensures the consistency of optimal solution under the worst-case measure with the risk-neutral measure.
\end{remark}

Similar to the studies in \cite{Ait07,Bakshi03,Black76}, further assumptions{,  $\lambda_j>0$,  $\mu_j<0$,  and  $\rho_j<0$,  are imposed}.  It follows immediately from \eqref{e17} that the optimal stock risk exposures $\beta^S_j$ are positive, whereas the optimal volatility risk exposures $\beta^V_j$ are negative. Obviously, each of the optimal exposures consists of a myopic component and a hedge component. The myopic components are constant in time and decrease as the corresponding ambiguity aversion parameters increase, whereas the hedge components are time-dependent and vanish as the investment horizon reaches the terminal time.
The worst-case measure $e^S_j$  and  $e^V_j$  in \eqref{opte} vary linearly with respect to the $j$-th factor of the volatility $\sqrt{V_j}\,$, but depend on both ambiguity aversion parameters $\phi^S_j$ and $\phi^V_j$ through function $H_j$. Thus, the model uncertainty can be thought of as the uncertainty about parameters $\lambda_j$ and $\mu_j$, which control the market prices of all risks.

To ensure that the optimal solutions are well-behaved \cite{Kraft05} and the function $J (t,x , v_1,v_2 )$  well-defined we present the following verification theorem.
\begin{proposition}\label{propa2}
The Radon-Nikodym derivative  under the worst-case measure 
$((e^S_1)^*,(e^S_2)^*,(e^V_1)^*,(e^V_2)^*)$
 is well-defined, thus the optimal portfolio is well-behaved, if
$$ (\phi_j^S)^2\frac{\lambda^2_j\sigma_j^2}{(\gamma+\phi^S_j)^2}+(\phi_j^V)^2\frac{\mu^2_j\sigma_j^2}{(\gamma+\phi^V_j)^2}\leq \kappa_j^2.$$
\end{proposition}
\begin{proof}
See Appendix \ref{A2}.
\end{proof}
\subsection{{\textbf{Incomplete Market Case}}}
When an investor has no access to derivative securities the volatility risks cannot be hedged perfectly so the market is incomplete. In this case, $\pi^i=0$ ($i=1, 2, 3$), the investor's wealth process $X_t$ is governed by
\begin{equation}\label{wealth2}
\begin{aligned}
\frac{dX(t)}{X(t)}=rdt+\sum_{j=1}^{2}\left[\pi^{S}(\lambda_jV_j-\sqrt{V_j}e^S_j
)dt+\pi^{S}\sqrt{V_j}d\widehat{W}_j(t)\right].
\end{aligned}
\end{equation}
Thus, the robust PDE in the incomplete market is
\begin{equation}\label{HJB02}
\begin{aligned}
\sup_{\pi^S}&\inf_{e^S_j,e^V_j}
\bigg\{
J_t+x\Big(r+\sum_{j=1}^{2}\pi^{S}(\lambda_jv_j-\sqrt{v_j}e^S_j)\Big)J_x+\frac{1}{2}x^2\sum_{j=1}^{2}(\pi^{S})^2v_jJ_{xx}
+\sum_{j=1}^{2}\big(\kappa_j(\theta_j-v_j)-\rho_j\sigma_j\sqrt{v_j}e^S_j\\
-&\sqrt{1-\rho_j^2}\sigma_j\sqrt{v_j}e^V_j\big)J_{v_j}+\sum_{j=1}^{2}\Big[\frac{1}{2}\sigma_j^2v_jJ_{v_jv_j}+\sigma_jv_jx\pi^S\rho_jJ_{v_jx}
+\frac{(e^S_j)^2}{2\Psi^S_j}+\frac{(e^V_j)^2}{2\Psi^V_j}\Big]\bigg\}=0.
\end{aligned}
\end{equation}
Our result is stated in the following proposition.
\begin{proposition}\label{prop2}
In an incomplete market, the indirect utility  of an ambiguity and risk averse investor is  \begin{equation}
J(t,x,v_1,v_2)=\frac{x^{1-\gamma}}{1-\gamma}\exp\left(\bar{H}_1(\tau)v_1+\bar{H}_2(\tau)v_2+\bar{h}(\tau)\right),
\end{equation}
where the functions $\bar{H}_1$, $\bar{H}_2$ and $\bar{h}$ are obtained from { \eqref{HJBA04}}.

{ Further, if} there exists only one volatility risk ($W_1$ or $W_2$),  or there are two equal risk factors ($W_1=W_2$), the optimal exposure to $W_j$ $(j=1\ \text{or}\  2)$ is
\begin{equation}
\begin{aligned}
\pi^S=\frac{\lambda_j}{\gamma+\phi^S_j}+\frac{(1-\gamma-\phi^S_j)\sigma_j\rho_j}{(1-\gamma)(\gamma+\phi^S_j)}\bar{H}_j(\tau),
\end{aligned}
\end{equation}
and the worst-case measures are
\begin{equation}\label{e+}
\begin{aligned}
e^S_j=\left(\frac{\lambda_j}{\gamma+\phi^S_j}+\frac{\sigma_j\rho_j \bar{H}_j(\tau)}{(1-\gamma)(\gamma+\phi^S_j)}\right)\phi^S_j\sqrt{v_j},
\ e^V_j=\frac{\sigma_j\sqrt{1-\rho_j^2} \bar{H}_j(\tau)}{1-\gamma}\phi^V_j\sqrt{v_j},
\end{aligned}
\end{equation}
where $\bar{H}_j$ are given by \eqref{HH2}.
\end{proposition}
\begin{proof}
See Appendix \ref{A3}.
\end{proof}
It should be emphasized that the optimal exposure  $\pi^S$ depends on the volatility ambiguity parameters $\phi^V_j$ through the function  $\bar{H}_j$ which
vanishes as the investment horizon decreases. Thus, the investment strategy of short-term investors in an incomplete market is relatively insensitive to volatility ambiguity. This is in sharp contrast to the case in the complete markets, where the volatility ambiguity parameter has an additional impact on the optimal strategy
through exposure $\beta^V_j$, which is absent in an incomplete market.
\subsection{{\textbf{Suboptimal Strategies and Utility Losses}}}
We first introduce the indirect utility function of the investor who follows an admissible suboptimal strategy.
\begin{definition}\label{def33}
For  an admissible suboptimal strategy $\Pi$, the indirect utility function is
\begin{equation}\label{max2}
J^{\Pi}(t,\mathbf{y})=\inf_{\mathbf{e}}\hskip-3pt
\Big(w^{{\mathbf{e}}}(t,\mathbf{y};\Pi)+\mathrm{E}^{\mathbb{P}^e}_{t,\mathbf{y}}\Big[\int_{t}^{T}\sum_{j=1}^2\frac{(e^S_j(s))^2}{2\Psi^S_j(s,\mathbf{Y})}+\frac{(e^V_j(s))^2}{2\Psi^V_j(s,\mathbf{Y})}ds
\Big]\Big).
\end{equation}
\end{definition}
\begin{remark}
It should be emphasized that by definition $J^\Pi(t,x,v_1,v_2)$ in \eqref{max2} is strictly less than $J(t,x,v_1,v_2)$ in \eqref{max},  the indirect utility function of the investor who follows an optimal strategy.
\end{remark}
The value function \eqref{max2} satisfies the following robust HJB PDE:
\begin{equation}\label{HJB22}
\begin{aligned}
\inf_{e^S_j,e^V_j}
\Big\{
&J^\Pi_t+x\big(r+\sum_{j=1}^{2}(\beta^{S}_j\lambda_jv_j-\beta^{S}_j\sqrt{v_j}e^S_j
+\beta^{V}_j\mu_jv_j-\beta^{V}_j\sqrt{v_j}e^V_j)\big)J^\Pi_x
+\frac{1}{2}x^2\sum_{j=1}^{2}[(\beta^{S}_j)^2+(\beta^{V}_j)^2]v_jJ^\Pi_{xx}\\
&+\frac{1}{2}\sum_{j=1}^{2}\sigma_j^2v_jJ^\Pi_{v_jv_j}
+\sum_{j=1}^{2}\big(\kappa_j(\theta_j-v_j)-\rho_j\sigma_j\sqrt{v_j}e^S_j-\sqrt{1-\rho_j^2}\sigma_j\sqrt{v_j}e^V_j\big)J^\Pi_{v_j}\\
&+\sum_{j=1}^{2}\Big[\sigma_jv_jx\big(\beta^S_j\rho_j+\beta^V_j\sqrt{1-\rho_j^2}\big)J^\Pi_{v_jx}
+\frac{(e^S_j)^2}{2\Psi^S_j}+\frac{(e^V_j)^2}{2\Psi^V_j}\Big]\Big\}=0.
\end{aligned}
\end{equation}
The solution of \eqref{HJB22} gives the general worst-case measures for suboptimal strategies:
\begin{equation}\label{es2}
\begin{aligned}
(e^S_j)^*=\Psi_j^S(x\beta^S_jJ_x+\rho_j\sigma_jJ_{v_j})\sqrt{v_j},\quad
\ (e^V_j)^*=\Psi_j^V(x\beta^V_jJ_x+\hskip -3pt\sqrt{1-\rho_j^2}\sigma_jJ_{v_j})\sqrt{v_j}.
\end{aligned}
\end{equation}
Similar to the studies in \cite{Bergen18,Escobar15,Escobar17,Flor14}, the wealth-equivalent utility loss $L^\Pi$ for a suboptimal strategy $\Pi$ is defined as the solution to
\begin{equation}\label{JJJ}
J(t,x(1-L^\Pi),v_1,v_2)=J^\Pi(t,x,v_1,v_2),
\end{equation}
where  $J$ is defined by \eqref{max}, and $J^\Pi$ by \eqref{max2}
such that the form of the indirect utility function is dictated by the form of the expected utility, $w^{{\mathbf{e}}}(t,\mathbf{y};\Pi)$, in \eqref{max1a}. Therefore, $J^\Pi$ is of exponential affine form, meaning economically the risk averse investor prefers a constant relative risk aversion (CRRA).
Thus,
$$L^\Pi=1-\exp\left\{\frac{1}{1-\gamma}\left[(H_1^\Pi-H_1)v_1+(H_2^\Pi-H_2)v_2+(h^\Pi-h)\right]\right\},$$
where $H_1$, $H^\Pi_1$, $H_2$, $H^\Pi_2$, $h$ and $h^\Pi$ are some functions as discussed below.

Here we consider three specific suboptimal strategies: $\Pi_1$, the investor ignores the uncertainty on the second volatility component; $\Pi_2$, the investor ignores the uncertainty about the first volatility component, and $\Pi_3$, the investor cannot trade derivatives. Therefore, the wealth-equivalent utility loss $L^\Pi$, which measures the percentage of wealth loss, consists of the loss due to the choice of  suboptimal strategy and { that} due to non-robustness from ignoring model uncertainty ($\Pi_1$ or $\Pi_2$) or market incompleteness ($\Pi_3$). We evaluate the indirect utility functions when an investor follows the sub-optimal strategies  $\Pi \in \{\Pi_1,\Pi_2,\Pi_3\}$. Because of the symmetry of $\Pi_1$ and $\Pi_2$, we  only discuss in detail for the cases associated with $\Pi_1$ and $\Pi_3$.

If an investor takes the strategy $\Pi_1$, all parameters describing the uncertainties of the stock risk and volatility risk associated with the first component of the volatility will disappear, i.e., $\widetilde{\phi}^S_1=\widetilde{\phi}^V_1=0$.
\begin{proposition}\label{prop3}
The indirect utility function of an investor who adopts strategy $\Pi_1$ is given by
\begin{equation}
J^{\Pi_1}(t,x,v_1,v_2)=\frac{x^{1-\gamma}}{1-\gamma}\exp\left(H^{\Pi_1}_1(\tau)v_1+H^{\Pi_1}_2(\tau)v_2+h^{\Pi_1}(\tau)\right),
\end{equation}
where  functions $H^{\Pi_1}_1$,  $H^{\Pi_1}_2$ and $h^{\Pi_1}$ are solved from { \eqref{e37}} by setting $\widetilde{\phi}^S_1=\widetilde{\phi}^V_1=0$,
\end{proposition}
\begin{proof}
See Appendix \ref{A4}.
\end{proof}
\begin{remark}
For an investor adopting strategy $\Pi_2$, the indirect utility function $J^{\Pi_2}$  can be obtained by setting $\widetilde{\phi}^S_2=\widetilde{\phi}^V_2=0$ in \eqref{e37}.
\end{remark}
We now consider the strategy $\Pi_3$ where the investor cannot trade derivatives, and report the utility loss under the multi-factor volatility model.
\begin{proposition}\label{prop4} The welfare loss from no derivative trading is strictly positive { such} that $L^{\Pi_3}>0$. In addition, in complete markets the welfare loss from ignoring the information about the ambiguity is also strictly positive, i.e., $L^{\Pi_1}>0$ and $L^{\Pi_2}>0$.
\end{proposition}
\begin{proof}
See Appendix \ref{A5}.
\end{proof}
\subsection{{\textbf{Asset Price with Jump Risks}}}
With jump risks, the asset price follows the dynamics
\begin{equation}\label{S2}
\begin{aligned}
\frac{dS(t)}{S(t-)}&=\left(r+\lambda_1V_1(t)+\lambda_2V_2(t)+j^S(\nu^\mathbb{P}-\nu^\mathbb{Q})(V_1(t)+V_2(t))\right) dt\\
&+\sum_{j=1}^2\sqrt{V_j(t)}dW_{j}(t)+j^S[dN(t)-\nu^\mathbb{P}(V_1(t)+V_2(t))dt],
\end{aligned}
\end{equation}
where $N(t)$ is a Poisson process independent of all Brownian motions with stochastic arrival intensity $\nu^\mathbb{P}(V_1(t)+V_2(t))$. Here the jump size of the Poisson process, $j^S > -1$, is set to be constant. Then the option price dynamics under multi-factor volatility is as follows
\begin{equation}\label{Oi2}
\begin{aligned}
dO^{(i)}(t)&=rO^{(i)}(t)dt+\sum_{j=1}^{2}\left(g^{(i)}_SS+\sigma_j\rho_jg^{(i)}_{V_j}\right)\left(\lambda_jV_jdt+\sqrt{V_j}dW_j(t)\right)\\
&+\sum_{j=1}^{2}\sigma_{j}\sqrt{1-\rho_j^2}g^{(i)}_{V_j}\left(\mu_jV_jdt+\sqrt{V_j}dZ_j(t)\right)+\triangle g^{(i)}\Big[(\nu^\mathbb{P}-\nu^\mathbb{Q})\sum_{j=1}^2V_j(t)+dN(t)-\nu^\mathbb{P}\sum_{j=1}^2V_j(t)dt\Big],\\
\end{aligned}
\end{equation}
where $\triangle g^{(i)}=g^{(i)}((1+j^S)S,V_1,V_2)-g^{(i)}(S,V_1,V_2)$.

{ Thus, we} obtain a new reference model combining multi-factor stock volatility with jump risks in asset price. To make the model more tractable, we assume that the investor is uncertain only about the distribution of Brownian motions. In other words, there is no ambiguity about $N_t$. One additional derivative is needed to hedge the added jump risks. For each perturbation process $\mathbf{e}$, the investor considers the following alternative model, where
the stock price is governed by
\begin{equation}
\begin{aligned}
\frac{dS(t)}{S(t-)}=&\big(r+\sum_{j=1}^{2}[\lambda_jV_j(t)-\sqrt{V_j(t)}e^S_j]+j^S(\nu^\mathbb{P}-\nu^\mathbb{Q})\sum_{j=1}^2V_j(t)\big) dt\\&+\sum_{j=1}^{2}\sqrt{V_j(t)}d\widehat{W}_{j}(t)+j^S\big[dN(t)-\nu^\mathbb{P}\sum_{j=1}^2V_j(t)dt\big],
\end{aligned}
\end{equation}
and at the same time, its variances follow \eqref{VE}
while the option prices of the stock satisfy
\begin{equation}
\begin{aligned}
dO^{(i)}(t)=&rO^{(i)}(t)dt+\hskip -2pt\sum_{j=1}^{2}\left(g^{(i)}_SS+\sigma_j\rho_jg^{(i)}_{V_j}\right)\left((\lambda_jV_j-e^S_j\sqrt{V_j})dt+\sqrt{V_j}d\widehat{W}_j(t)\right)+\sum_{j=1}^{2}\sigma_{j}\sqrt{1-\rho_j^2}g^{(i)}_{V_j}\\
&\left((\mu_jV_j-e^V_j\sqrt{V_j})dt+\sqrt{V_j}d\widehat{Z}_j(t)\right)+\triangle g^{(i)}\big[(\nu^\mathbb{P}-\nu^\mathbb{Q})\sum_{j=1}^2V_j(t)+dN(t)-\nu^\mathbb{P}\sum_{j=1}^2V_j(t)dt\big],\\
\end{aligned}
\end{equation}
where $\triangle g^{(i)}=g^{(i)}((1+{j^S})S,V_1,V_2{)}-g^{(i)}(S,V_1,V_2{)}, \, i = 1,\cdots,4$.

The value function \eqref{max} for this case satisfies the robust HJB PDE:
\begin{equation}\label{HJBJ}
\begin{aligned}
&\sup_{\beta^S_j,\beta^V_j,\beta^N}\inf_{e^S_j,e^V_j}
\bigg\{
J_t+x\big(r+\sum_{j=1}^{2}(\beta^{S}_j\lambda_jv_j-\beta^{S}_j\sqrt{v_j}e^S_j
+\beta^{V}_j\mu_jv_j-\beta^{V}_j\sqrt{v_j}e^V_j)-\beta^Nj^S\nu^\mathbb{Q}(v_1+v_2)\big)J_x\\
&
+\frac{1}{2}x^2\sum_{j=1}^{2}[(\beta^{S}_j)^2+(\beta^{V}_j)^2]v_jJ_{xx}+\frac{1}{2}\sum_{j=1}^{2}\sigma_j^2v_jJ_{v_jv_j}+\sum_{j=1}^{2}\sigma_jv_jx\big(\beta^S_j\rho_j+\beta^V_j\sqrt{1-\rho_j^2}\big)J_{v_jx}+\sum_{j=1}^{2}\frac{(e^S_j)^2}{2\Psi^S_j}+\frac{(e^V_j)^2}{2\Psi^V_j}\\
&+\sum_{j=1}^{2}\big(\kappa_j(\theta_j-v_j)-\rho_j\sigma_j\sqrt{v_j}e^S_j-\sqrt{1-\rho_j^2}\sigma_j\sqrt{v_j}e^V_j\big)J_{v_j}+\nu^\mathbb{P}(v_1+v_2)\triangle J\bigg\}=0,
\end{aligned}
\end{equation}
where $\triangle J=J(t,x(1+\beta^N{j^S}),v_1,v_2{)}-J(t,x,v_1,v_2{)}$.
\begin{proposition}\label{propJ1}
In a complete market, the indirect utility  of an ambiguity and risk averse investor is
\begin{equation}
J(t,x,v_1,v_2)=\frac{x^{1-\gamma}}{1-\gamma}\exp\left(C_1(\tau)v_1+C_2(\tau)v_2+c(\tau)\right),
\end{equation}
where  
\begin{equation}\label{Hh12}
\begin{aligned}
C_j(\tau)&=\frac{2s_j(1-e^{-d_j\tau})}{2d_j+(a_j+d_j)(e^{-d_j\tau}-1)}, \ j =1, 2,\\
c(\tau)&=(1-\gamma)r\tau-\hskip -3pt\sum_{j=1}^{2}\kappa_j\theta_j\hskip -3pt\left[\hskip -2pt\frac{a_j+d_j}{2b_j}\tau+\frac{1}{b_j}\ln\hskip -3pt\left(\frac{e^{-d_j\tau(a_j+d_j)}-a_j+d_j}{2d_j}\right)\right],
\end{aligned}
\end{equation}
with $d_j=\sqrt{a_j^2-4b_js_j}$,
and
\begin{eqnarray*}
\begin{aligned}
a_j=&-\kappa_j+\frac{\lambda_1(1-\gamma-\phi_j^S)\sigma_j\rho_j}{\gamma+\phi_j^S}+\frac{\lambda_2(1-\gamma-\phi_j^V)\sigma_j}{\gamma+\phi_j^V}\sqrt{1-\rho_j^2},\\
b_j=&\frac{\sigma_j^2}{2}-\frac{\phi^S_j\rho_j^2\sigma_j^2}{2(1-\gamma)}-\frac{\phi^V_j(1-\rho_j^2)\sigma_j^2}{2(1-\gamma)}+\frac{(1-\gamma-\phi_j^S)\sigma^2_j\rho^2_j}{2(1-\gamma)(\gamma+\phi_j^S)}+\frac{(1-\gamma-\phi_j^V)^2\sigma^2_j(1-\rho_j^2)}{2(1-\gamma)(\gamma+\phi_j^V)},\\
s_j=&\frac{(1-\gamma)\lambda_1^2}{2(\gamma+\phi^S_j)}+\frac{(1-\gamma)\lambda_2^2}{2(\gamma+\phi^V_j)}+\left((\frac{\nu^\mathbb{P}}{\nu^\mathbb{Q}})^{1/\gamma}\gamma\nu^\mathbb{Q}-(\gamma-1)\nu^\mathbb{Q}-\nu^\mathbb{P}\right).
\end{aligned}
\end{eqnarray*}
The optimal exposures to the risk factors { $N(t)$, $W_j$, $Z_j$ $(j=1,2)$} are $\beta^N=\frac{1}{j^S}\left(\left(\frac{\nu^\mathbb{P}}{\nu^\mathbb{Q}}\right)^{1/\gamma}-1\right)$,
\begin{equation}\label{e99}
\begin{aligned}
\beta^S_j&=\frac{\lambda_j}{\gamma+\phi^S_j}+\frac{(1-\gamma-\phi^S_j)\sigma_j\rho_j}{(1-\gamma)(\gamma+\phi^S_j)}C_j(\tau),
\quad \beta^V_j=\frac{\mu_j}{\gamma+\phi^V_j}+\frac{(1-\gamma-\phi^V_j)\sigma_j\sqrt{1-\rho_j^2}}{(1-\gamma)(\gamma+\phi^V_j)}C_j(\tau). 
\end{aligned}
\end{equation}
The worst-case measure are
\begin{equation}
e^S_j=\left(\frac{\lambda_j}{\gamma+\phi^S_j}+\frac{\sigma_j\rho_j C_j(\tau)}{(1-\gamma)(\gamma+\phi^S_j)}\right)\phi^S_j\sqrt{v_j},\quad e^V_j=\left(\frac{\mu_j}{\gamma+\phi^V_j}+\frac{\sigma_j\sqrt{1-\rho_j^2} C_j(\tau)}{(1-\gamma)(\gamma+\phi^V_j)}\right)\phi^V_j\sqrt{v_j}.
\end{equation}
\end{proposition}
\begin{proof}
The proof, similar to that of Proposition \ref{prop1}, is omitted here with the full proof at arXiv:1910.06872.
\end{proof}
\begin{proposition}\label{propJ2}
In an incomplete market, the indirect utility of an ambiguity and risk averse investor is \begin{equation}
J=\frac{x^{1-\gamma}}{1-\gamma}\exp\left(\bar{C}_1(\tau)v_1+\bar{C}_2(\tau)v_2+\bar{c}(\tau)\right),
\end{equation}
where functions $\bar{C}_j$ and $\bar{c}$ can be derived as similar forms to those in Proposition \ref{prop2}. The optimal exposures to the stock risk factors $W_j$ $(j=1,2)$ are
\begin{equation}
\begin{aligned}
\beta^S_j&=\frac{\lambda_j}{\gamma+\phi^S_j}+\frac{(1-\gamma-\phi^S_j)\sigma_j\rho_j}{(1-\gamma)(\gamma+\phi^S_j)}\bar{C}_j(\tau).
\end{aligned}
\end{equation}
The worst-case measures are
\begin{equation}
\begin{aligned}
e^S_j=\left(\beta^S_j+\frac{\sigma_j\rho_j \bar{C}_j(\tau)}{1-\gamma}\right)\phi^S_j\sqrt{v_j},
\quad e^V_j=\frac{\sigma_j\sqrt{1-\rho_j^2} \bar{C}_j(\tau)}{1-\gamma}\phi^V_j\sqrt{v_j}.
\end{aligned}
\end{equation}
\end{proposition}
\begin{proof}
The proof,  similar to that of Proposition \ref{prop2}, is omitted here with the full proof at arXiv:1910.06872.
\end{proof}
\begin{corollary}\label{cor31}
The utility loss from ignoring jump risk is strictly positive if $\nu^\mathbb{P}\neq \nu^\mathbb{Q}$ and is zero otherwise. In other words,  $\nu^\mathbb{P}= \nu^\mathbb{Q}$ implies that the suboptimal strategy of ignoring jump risk becomes optimal, and thus nullifies the welfare loss.
\end{corollary}
\section{Correlated Volatility Processes}\label{Sec 4}
In this section, we extend our model to the case where  the two volatility { factors} are correlated.  This is an often occurring phenomenon which has been verified in the literature \cite{Grzelak09,Grzelak11,Grzelak12,Yang17}. Assuming $\langle dW_1(t),dW_2(t)\rangle=\rho dt$, we obtain a new robust PDE with respect to $X$, $V_1$ and $V_2$ after solving the robust optimization problem \eqref{max}. However, we encounter difficulty in solving the new robust PDE, as the correlation between $W_1$ and $W_2$ adds a non-affine term $\sqrt{V_1V_2}$:
\begin{equation}\label{HJBC}
\begin{aligned}
\sup_{\beta^S_j,\beta^V_j}&\inf_{e^S_j,e^V_j}
\bigg\{
J_t+x\big(r+\sum_{j=1}^{2}(\beta^{S}_j\lambda_jv_j-\beta^{S}_j\sqrt{v_j}e^S_j
+\beta^{V}_j\mu_jv_j-\beta^{V}_j\sqrt{v_j}e^V_j)\big)J_x+\frac{1}{2}
\sum_{j=1}^{2}\sigma_j^2v_jJ_{v_jv_j}
\\
&
+\sum_{j=1}^{2}\big(\kappa_j(\theta_j-v_j)-\rho_j\sigma_j
\sqrt{v_j}e^S_j-\sqrt{1-\rho_j^2}
\sigma_j\sqrt{v_j}e^V_j\big)J_{v_j}+\frac{1}{2}x^2\sum_{j=1}^{2}[(\beta^{S}_j)^2+(\beta^{V}_j)^2]v_jJ_{xx}\\
&+\sum_{j=1}^{2}\sigma_jv_jx\big(\beta^S_j\rho_j+\beta^V_j \sqrt{1-\rho_j^2}\big)J_{v_jx}
+\rho\rho_1\rho_2\sigma_1\sigma_2\sqrt{v_1v_2}J_{v_1v_2}+\sum_{j=1}^{2}\frac{(e^S_j)^2}{2\Psi^S_j}+\frac{(e^V_j)^2}{2\Psi^V_j}\bigg\}=0.
\end{aligned}
\end{equation}
Inspired by the ideas of Grzelak et al. \cite{Grzelak09,Grzelak11,Grzelak12}, we develop an approximation method to solve the PDE, hence the optimal problem. For easiness of the readers, we recap some of the necessary steps in  \cite{Grzelak09} here.
For the mean-reverting variance processes in \eqref{V0} we have {$V_j(t)=c_j(t){\chi^2}(d_{j},\lambda_j(t)), \ j =1, 2$,}
where $\chi^2(d_j,\lambda_j(t))$ is a non-central chi-squared random variable with the degree of freedom parameter  $d_j$ and non-centrality parameter $\lambda_j(t)$, and
$$c_j(t)=\frac{1}{4\kappa_j}\sigma_j^2(1-e^{-\kappa_j(t)}),
\quad
d_j=\frac{4\kappa_j\theta_j}{\sigma_j^2},
\quad
\lambda_j(t)=\frac{4\kappa_je^{-\kappa_jt}V_j(0)}{\sigma_j^2(1-e^{-\kappa_jt})}.
$$
The expectation and variance of $\sqrt {V_j}$ can be found as (Lemma 3.1 \cite{Grzelak09})
\begin{equation}\label{exp1}
\begin{aligned}
\mathrm{E}\hskip -3pt\left(\hskip -3pt\sqrt{V_j(t)}\right) & =  \sqrt{2c_j(t)}e^{-\lambda_j(t)/2}\sum_{k=0}^{\infty}\frac{1}{k!}\left(\frac{\lambda_j(t)}{2}\right)^{k_j}\frac{\Gamma(\frac{d_j+1}{2}+k)}{\Gamma(\frac{d_j}{2}+k)}, \\
 \mathrm{Var}\hskip -3pt\left(\hskip -3pt\sqrt{V_j(t)}\right)& =  c_j(t)(d_j+\lambda_j(t))-2c_j(t)e^{-\lambda_j(t)}\hskip -3pt\Big[\hskip -3pt\sum_{k=0}^{\infty}\frac{1}{k!}\Big(\frac{\lambda_j(t)}{2}\Big)^{k_j}\hskip -4pt\frac{\Gamma(\frac{d_j+1}{2}+k)}{\Gamma(\frac{d_j}{2}+k)}\Big]^2,
\end{aligned}
\end{equation}
where
$\displaystyle \Gamma(z)=\int_{0}^{\infty}t^{z-1}e^{-t}dt$.

{
The dynamics $d\sqrt{V_j}$, $j=1$, $2$, can be approximated by (Lemma 3.4 \cite{Grzelak09}):
\begin{equation}\label{Uj}
dU_j(t)=\mu_j^U(t)dt+\psi^U_j(t)\hskip -3pt\left(\rho_j dW_{j}(t)+\hskip -3pt\sqrt{1-\rho_j^2} dZ_j(t)\right), \  U_j(0)=\sqrt{V_j(0)}>0.
\end{equation}
The deterministic time-dependent drift $\mu_j^U(t)$ and volatility $\psi_j^U(t)$ are} given by
\begin{equation*}
\begin{aligned}
\mu^U_j(t)=&\frac{1}{2\sqrt{2}}\frac{\Gamma(\frac{d_j+1}{2})}{\sqrt{c_j(t)}}\bigg(\frac{1}{2}\sigma_j^2e^{-\kappa_jt}\widetilde{F}\left(-\frac{1}{2},\frac{d}{2},-\frac{\lambda_j(t)}{2}\right)
-\frac{4c_j(t)\kappa_j^2e^{\kappa_jt}V_j(0)}{\sigma_j^2(e^{\kappa_jt}-1)^2}\widetilde{F}\left(\frac{1}{2},\frac{d+1}{2},-\frac{\lambda_j(t)}{2}\right)\bigg), \\
\psi^U_j(t)=&\left(\frac{\sigma_j^2e^{-\kappa_jt}}{4}(d_j+\lambda_j(t))
-\frac{4c_j(t)\kappa_j(t)e^{\kappa_jt}V_j(0)}{\sigma_j^2(e^{\sigma_jt}-1)^2}
-2\mathrm{E}(V_j(t))\mu^U_j(t)\right)^{1/2},
\end{aligned}
\end{equation*}
where $\mathrm{E}(V_j(t))$ is given by \eqref{exp1} and the regularized hypergeometric function $\widetilde{F}(a;b;z)=F(a;b;c)/\Gamma(b)$.

Note that the accuracy and convergency of the approximation in \eqref{Uj} { are} well tested in \cite{Grzelak09}, and the order of this type of approximation and the error estimates can be found in \cite{Johnson94}. Hence,  we can construct an approximation, with confidence, by introducing $Y(t)=U_1(t)U_2(t)\approx \sqrt{V_1(t)V_2(t)}$ such that
\begin{equation}\label{Y}
\begin{aligned}
dY(t)=&\mu^Ydt+U_1(t)\psi^U_2(t)\left(\rho_2 dW_{2}(t)+\sqrt{1-\rho_2^2} dZ_2(t)\right)+U_2(t)\psi^U_1(t)\left(\rho_1 dW_{1}(t)+\sqrt{1-\rho_1^2} dZ_1(t)\right),
\end{aligned}
\end{equation}
where $\mu^Y=\mu^U_1(t)U_2(t)+\mu^U_2(t)U_1(t)+\rho\psi^U_1(t)\psi^U_2(t).$

Equations \eqref{S1} - \eqref{Oi}, {\eqref{Uj} and \eqref{Y}} form the new reference model.  Allowing the perturbation process $\mathbf{e}$, the corresponding alternative model consists of Equations \eqref{SE}-\eqref{OE} with the following affine correction processes
\begin{equation*}\label{Ye}
\begin{aligned}
dU_j(t) =&\mu_j^U(t)dt+\psi^U_j(t)\left[\rho_j (d\widetilde{W}_{j}(t)-e^S_j(t))+\sqrt{1-\rho_j^2} (d\widetilde{Z}_j(t)-e^V_j(t))\right],\ U_j(0)=\sqrt{V_j(0)}>0,\\
dY(t)=&\mu^Ydt+U_1(t)\psi^U_2(t)\left[\rho_2 (d\widetilde{W}_{2}(t)-e^S_2(t))+\sqrt{1-\rho_2^2} (d\widetilde{Z}_2(t)-e^V_2(t))\right]\\
&+U_2(t)\psi^U_1(t)\left[\rho_1 (d\widetilde{W}_{1}(t)-e^S_1(t))+\sqrt{1-\rho_1^2} (d\widetilde{Z}_1(t)-e^V_1(t))\right],\ Y(0)=\sqrt{V_1(0)V_2(0)}>0.
\end{aligned}
\end{equation*}
We now introduce a new process $\textbf{R}(s)=(X(s),V_1(s),V_2(s),U_1(s),U_2(s),Y(s))$.  The expected utility achieved by a trading strategy $\Pi$ is given by
\begin{equation}\label{max1}
w^\mathbf{e}(t,\textbf{r};\Pi)=\frac{1}{1-\gamma}\mathrm{E}^{\mathbb{P}^\mathbf{e}}_{t,\textbf{r}}\left[(X_T)^{1-\gamma}\right],
\end{equation}
where  $\textbf{r}=(x,v_1,v_2,u_1,u_2,y)$ denotes the value of $\textbf{R}(t)$ at time $t$.
The indirect utility  of the investor is
\begin{equation}\label{max3}
J(t,\textbf{r})=\sup_{\Pi}\inf_{\mathbf{e}}\hskip -3pt
\Big(w^\mathbf{e}(t,\textbf{r};\Pi)+\mathrm{E}^{\mathbb{P}^\mathbf{e}}_{t,\textbf{r}}\Big[\int_{t}^{T}\sum_{j=1}^2\frac{(e^S_j(s))^2}{2\Psi^S_j(s,Y)}+\frac{(e^V_j(s))^2}{2\Psi^V_j(s,Y)}ds
\Big]\Big).
\end{equation}
In the complete market, the value function \eqref{max3} satisfies the robust HJB PDE \eqref{HJBC} as in Appendix \ref{sec4}. { Assuming that the value function $J$ is of the following affine-form
\begin{equation}\label{JC}
J(t,x,v_1,v_2,u_1,u_2,y)=\frac{x^{1-\gamma}}{1-\gamma}\exp\Big(\sum_{j=1}^{2}[H^V_j(T-t)v_j+H^U_j(T-t)u_j]+H^Y(T-t)y+\hat h\Big),
\end{equation}
solving the} optimal problem \eqref{HJBC} with respect to $e^S_1$, $e^S_2$, $e^V_1$, $e^V_2$, we obtain
\begin{equation}\label{ec}
\begin{aligned}
(e^S_1)^*&=\Psi^S_1(x\beta^S_1u_1J_x+\rho_1\sigma_1u_1J_{v_1}+\rho_1\psi^U_1J_{u_1}-\rho_1\psi^U_1u_2J_y),\\
(e^S_2)^*&=\Psi^S_2(x\beta^S_2u_2J_x+\rho_2\sigma_2u_2J_{v_2}+\rho_1\psi^U_1J_{u_2}-\rho_2\psi^U_2u_1J_y),\\
(e^V_1)^*&=\Psi^V_1(x\beta^V_1u_1J_x+\sqrt{1-\rho_1^2}\sigma_1u_1J_{v_1}-\sqrt{1-\rho_1^2}\psi_1^UJ_{u_1}+\sqrt{1-\rho_1^2}\psi_1^Uu_2J_y),\\
(e^V_2)^*&=\Psi^V_2(x\beta^V_2u_2J_x+\sqrt{1-\rho_2^2}\sigma_2u_2J_{v_2}-\sqrt{1-\rho_2^2}\psi^U_2J_{u_2}+\sqrt{1-\rho_2^2}\psi^U_2u_1J_y).
\end{aligned}
\end{equation}
Substituting \eqref{ec} into \eqref{HJBC} and choosing $ \Psi^S_j=\frac{\phi^S_j}{(1-\gamma)J}$ and $\Psi^V_j=\frac{\phi^V_j}{(1-\gamma)J}$,
we further derive
\begin{equation}\label{bc1}
\begin{aligned}
(\beta^S_1)^*&=\frac{1}{\phi^S_1-\gamma}
\big[\lambda_1+\sigma_1\rho_1H^{V_1}+\frac{1}{u_1}\rho\rho_2\psi^U_2H^{U_2}+(\rho\rho_2\psi^U_2+\rho_1\frac{u_2}{u_1}\psi^U_1)H^Y\\
&\quad-\frac{\psi_1^S}{1-\gamma}(\rho\sigma_1H^{V_1}+\frac{1}{u_1}\rho_1\psi^U_1H^{U_1}-\frac{u_2}{u_1}\rho\psi^U_1H^Y)\big]\\
&\mathbf{=:}a^1_1(t)+a^1_2(t)\frac{1}{u_1}+a^1_3(t)\frac{u_2}{u_1},\\
(\beta^S_2)^*&=\frac{1}{\phi^S_2-\gamma}
\big[\lambda_2+\sigma_2\rho_2H^{V_2}+\frac{1}{u_2}\rho\rho_1\psi^U_1H^{U_1}+(\rho\rho_1\psi^U_1+\rho_2\frac{u_1}{u_2}\psi^U_1)H^Y\\
&\quad-\frac{\psi^S_2}{1-\gamma}(\rho\sigma_1H^{V_2}+\frac{1}{u_2}\rho_2\psi^U_2H^{U_2}-\frac{u_1}{u_2}\rho\psi^U_2H^Y)\big]\\
&\mathbf{=:}a^2_2(t)+a^2_2(t)\frac{1}{u_2}+a^2_3(t)\frac{u_1}{u_2},\\
(\beta^V_1)^*&=\frac{1}{\phi^V_1-\gamma}\big[\mu_1+\sqrt{1-\rho_1^2}\psi^U_1H^Y-\frac{\psi_1^V\sqrt{1-\rho_1^2}}{1-\gamma}(\sigma_1H^{V_1}-\psi^U_1\frac{1}{u_1}H^{U_1}\\
&-\psi^U_1\frac{u_2}{u_1}H^Y)\big]\mathbf{=:}b^1_1(t)+b^1_2(t)\frac{1}{u_1}+b^1_3(t)\frac{u_2}{u_1},\\
(\beta^V_2)^*&=\frac{1}{\phi^V_2-\gamma}\big[\mu_2+\sqrt{1-\rho_2^2}\psi^U_2H^Y-\frac{\psi_2^V\sqrt{1-\rho_2^2}}{1-\gamma}(\sigma_2H^{V_2}-\psi^U_2\frac{1}{u_2}H^{U_2}\\
&-\psi^U_2\frac{u_1}{u_2}H^Y)\big]\mathbf{=:}b^2_1(t)+b^2_2(t)\frac{1}{u_2}+b^2_3(t)\frac{u_1}{u_2}.\\
\end{aligned}
\end{equation}
Combining \eqref{JC} - \eqref{bc1},  we obtain for $j=1,2$,
\begin{equation}\label{ec2}
\left\{
\begin{aligned}
(e^S_j)^*&=g^j_1(t)u_1+g^j_2(t)u_2+g^j_3(t),\\
(e^V_j)^*&=k^j_1(t)u_1+k^j_2(t)u_2+k^j_3(t)
\end{aligned}
\right.
\end{equation}
Plugging \eqref{JC} - \eqref{ec2} into the robust HJB \eqref{HJBC}, and collecting terms with respect to $v_j$,  $u_j$ and $y$ we have the following simpler HJB:
\begin{equation}
0=J_t+p^{V}_1(t)v_1J+p^{V}_2(t)v_2J+p^{U}_1(t)u_1J+p^{U}_2(t)u_2J+p^Y(t)yJ+C(t),
\end{equation}
which leads to the following ODEs for the solutions of $H^{V}_J$, $H^{U}_j$, $H^Y$ and $\hat h$:
\begin{equation}\label{fode}
(H^{V}_j)'=-p^{V}_j(t), \ (H^{U}_j)'=-p^{U}_j(t), \ (H^{Y})'=-p^{Y}(t), \ \hat h'=-C(t).
\end{equation}
Thus, the optimal exposures and the worst-case measures can be derived analytically.

{ It is worth noting} that the correlation between volatility processes  affects both worst-case measures and optimal exposures. As mentioned before the optimal exposures to the stock and additional multi-factor volatility risks consist of myopic and hedge components. Each hedge component is dependent on the correlated volatility structure.
The worst-case measures to the stock price and volatility also depend on the correlated volatility structure.

\section{Numerical Experiments}

In this section, numerical experiments are carried out to examine the behavior of the optimal portfolio and the utility losses. In our calculations,  parameter  are set as those in some of our references: the investment horizon is $T=10$ (years), the risk-free interest rate  $r=0.05$ and the risk aversion parameters $\gamma=4$ \cite{Flor14, Escobar15};  $j^S=-15\%$, $\nu^\mathbb{P}=0.1$ and $\nu^\mathbb{Q}=0.3$ \cite{Branger13, Escobar15};    $V_1(0)=0.2^2$, $\kappa_1=3$, $\theta_1=0.1^2$, $\rho_1=-0.7$,
$\sigma_1=0.25$, $\lambda_1=3$, $\mu_1=-3$, $V_2(0)=0.01^2$, $\kappa_2=3.5$, $\theta_2=0.2^2$, $\rho_2=-0.3$,
$\sigma_2=0.01$, $\lambda_2=2$, $\mu_2=-3$ \cite{Liu03}.

The parameters $\phi^S_1$, $\phi^S_2$, $\phi^V_1$ and $\phi^V_2$, { which are defined in Section 2,  describe the investor's} preference for robustness. Anderson et al. \cite{Anderson03} argue that these parameters may be chosen in such a way that it could be difficult to distinguish the reference model from the worst case model based on a time series of { finite} length. Hence, we discuss the detection-error probability $\varepsilon_T$ as a function of $\phi^S_1,\phi^S_2,\phi^V_1$, and $\phi^V_2$.
Define the Radon-Nikodym derivatives
$\mathcal{Z}_1(t)=\mathrm{E}^\mathbb{P}[\frac{d\mathbb{P}^e}{d\mathbb{P}}|\mathcal{F}_t]$
and
$\mathcal{Z}_2(t)=\mathrm{E}^{\mathbb{P}^e}[\frac{d\mathbb{P}}{d\mathbb{P}^e}|\mathcal{F}_t]$,
and consider their logarithms:
\begin{equation*}
\begin{aligned}
\xi_1(t)&=\ln{\mathcal{Z}_1(t)}=-\int_{0}^{t}\sum_{j=1}^{2}\left[\frac{1}{2}\left((e^S_j(\tau))^2+(e^V_j(\tau))^2\right)d\tau+e^S_j(\tau)dW_j(\tau)+e^V_j(\tau)dZ_j(\tau)\right],\\
\xi_2(t)&=\ln{\mathcal{Z}_2(t)}=\int_{0}^{t}\sum_{j=1}^{2}\left[\frac{1}{2}\left((e^S_j(\tau))^2+(e^V_j(\tau))^2\right)d\tau+e^S_j(\tau)dW_j(\tau)+e^V_j(\tau)dZ_j(\tau)\right].
\end{aligned}
\end{equation*}
Based on a sample of finite length $T$, { clearly, the decision maker will  discard the reference model mistaken} for the worst case model
if $\xi_1(T)> 0$. On the other hand, if the worst case model is true, then it will be rejected erroneously if $\xi_2(T)> 0$ (or $\xi_1(T)< 0$). Thus, { we define the detection-error probability as follows}
$$\varepsilon_T(\phi^S_1,\phi^S_2,\phi^V_1,\phi^V_2)=\frac{1}{2}\mathrm{P}(\xi_1(T)> 0|\mathbb{P},\mathcal{F}_0)+\frac{1}{2}\mathrm{P}(\xi_1(T)< 0|\mathbb{P}^e,\mathcal{F}_0), $$
which can be calculated explicitly {using the Fourier inversion method, more details are given in Appendix~\ref{A6}.}

{ As shown in Figures \ref{fig1a} - \ref{fig1bb},  the detection-error probability  varies inversely with the ambiguity-aversion parameters in both complete and incomplete markets, i.e., the larger the ambiguity-aversion parameters, the less the probability of the investor making mistakes.  Overall, the detection-error probability changes with respect to the volatility factors $V_1$ and $V_2$ at different rates.  Furthermore,  it also exhibits different sensitivities to the ambiguity-aversion level for asset price and that for volatility.}
\begin{figure}
\centering
\subfigure[$\varepsilon$ vs. $(\phi_1^S, \phi_1^V)$ in complete markets]
{\label{fig1a}
\includegraphics[width=0.4\textwidth]{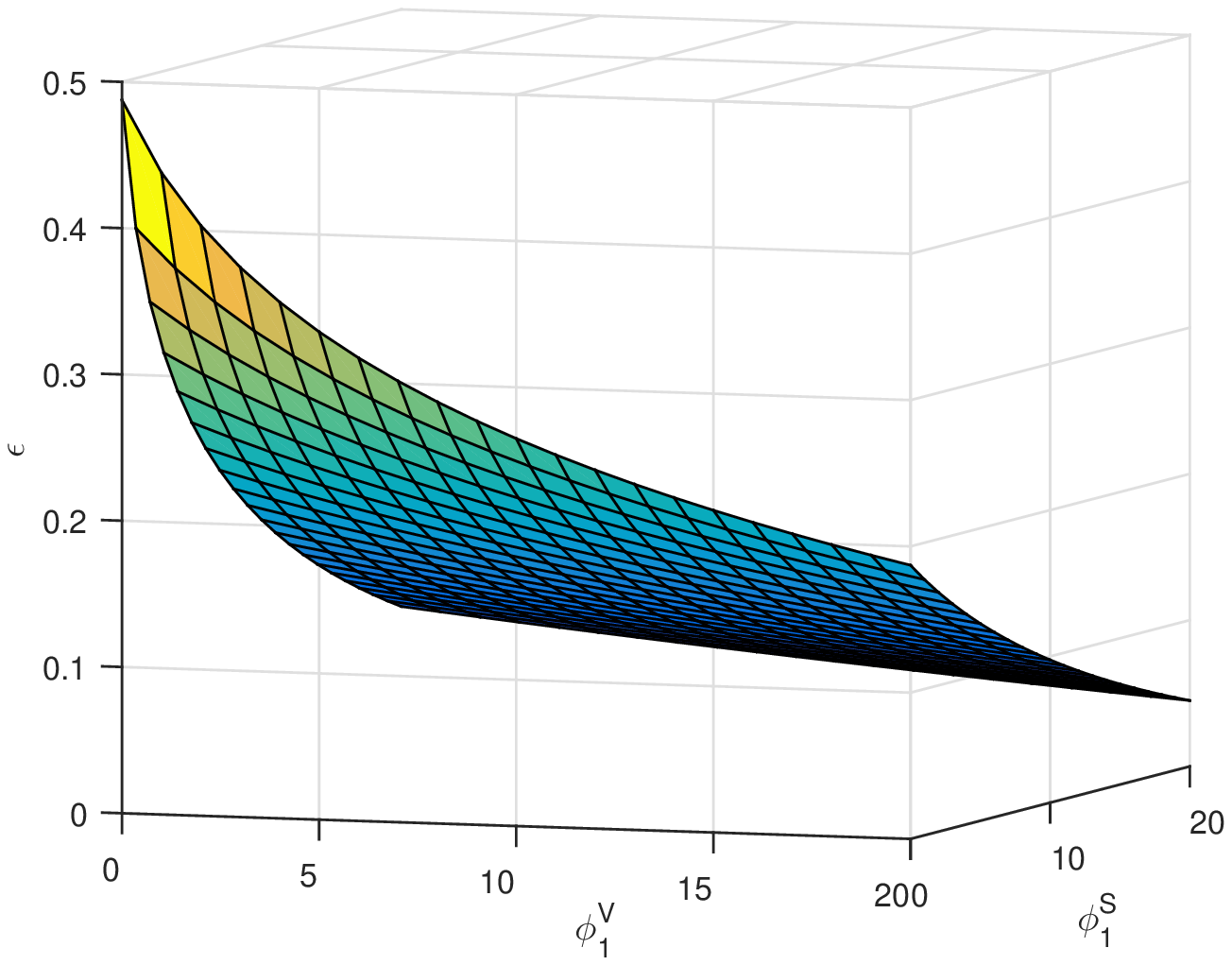}}
\subfigure[$\varepsilon$ vs. $(\phi_2^S, \phi_2^V)$ in complete markets]
{\label{fig1b}
\includegraphics[width=0.4\textwidth]{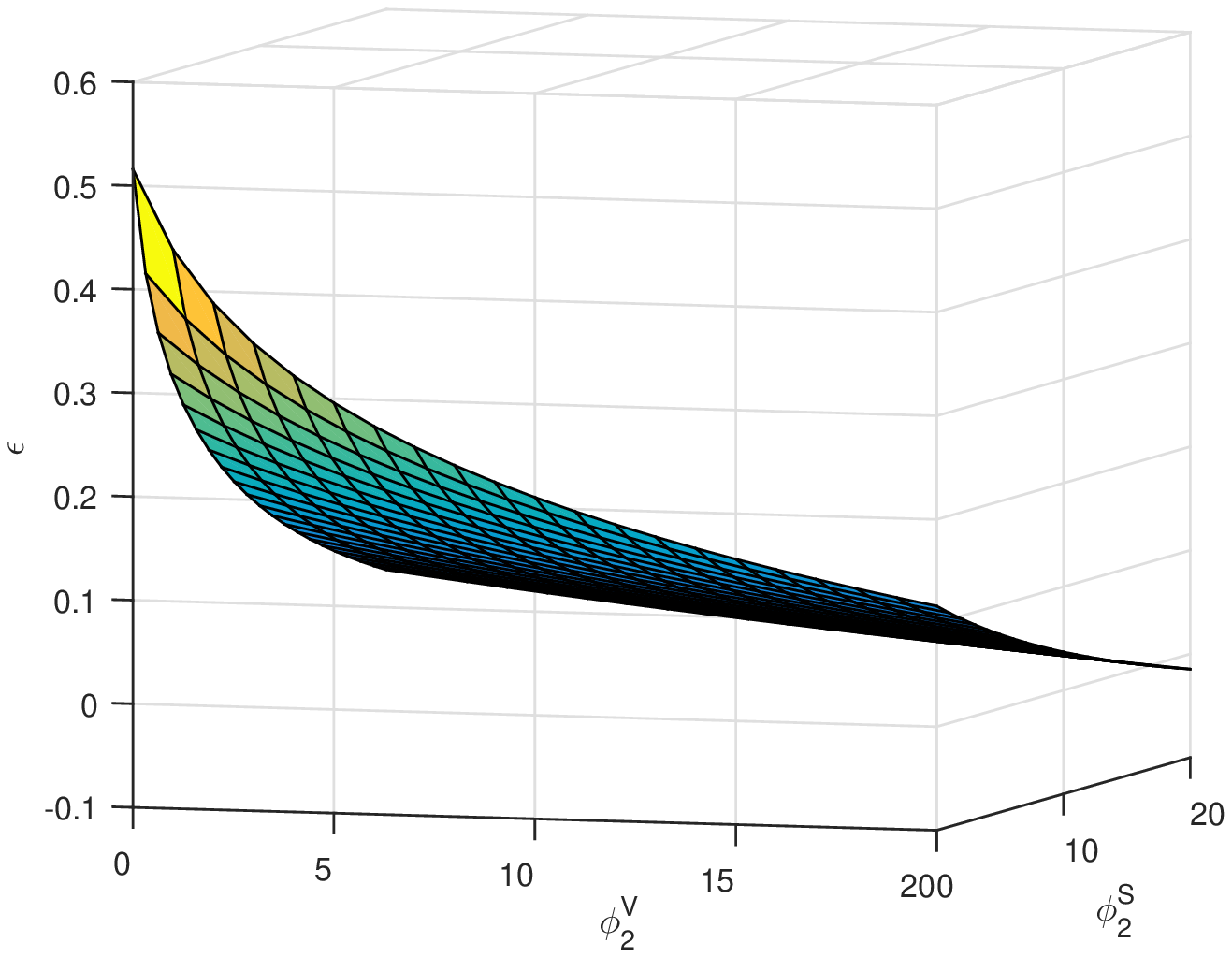}}
\subfigure[$\varepsilon$ vs. $(\phi_1^S, \phi_1^V)$ in incomplete markets]
{\label{fig1aa}
\includegraphics[width=0.4\textwidth]{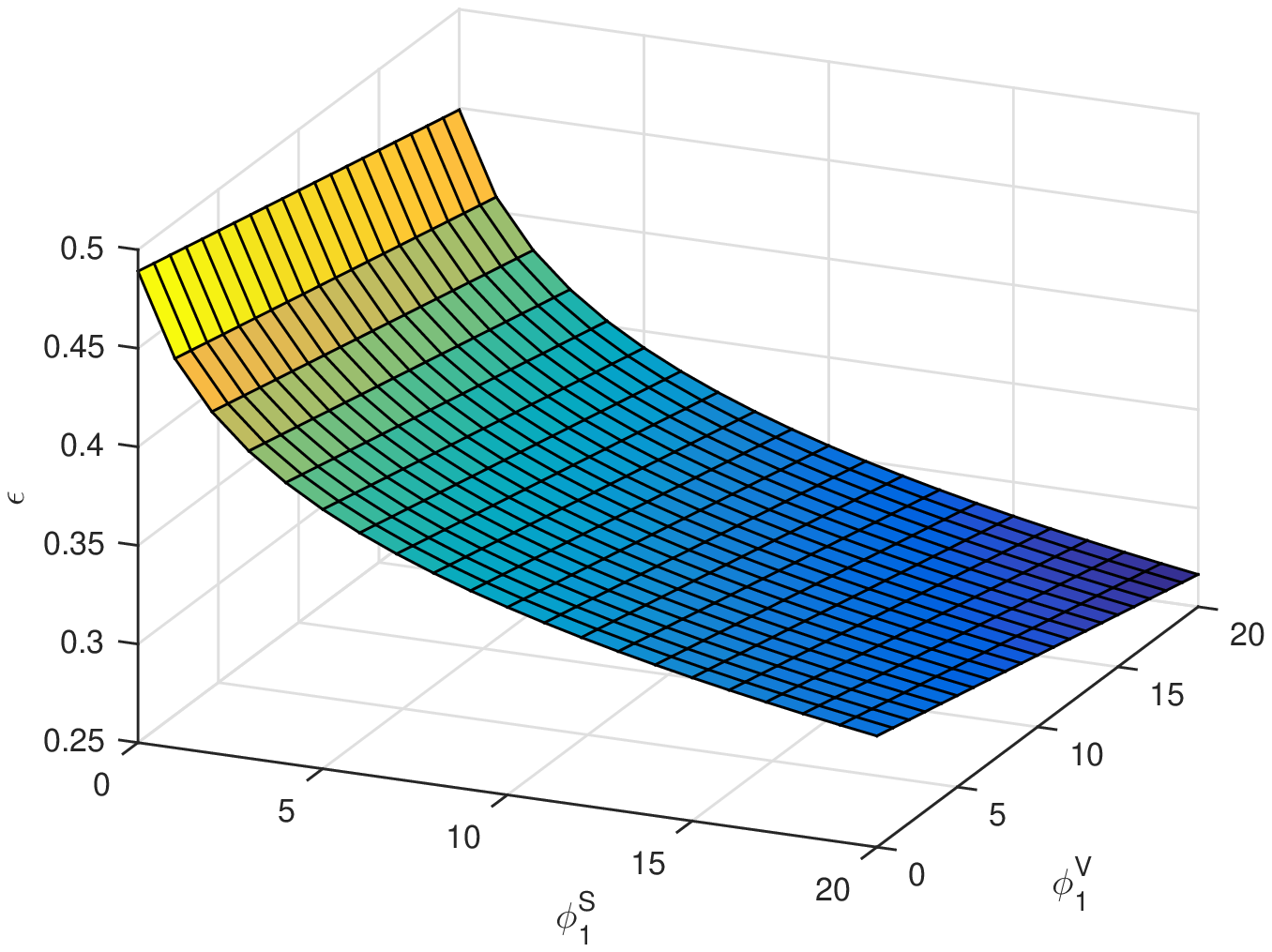}}
\subfigure[$\varepsilon$ vs. $(\phi_2^S, \phi_2^V)$ in incomplete markets]
{\label{fig1bb}\includegraphics[width=0.48\textwidth]{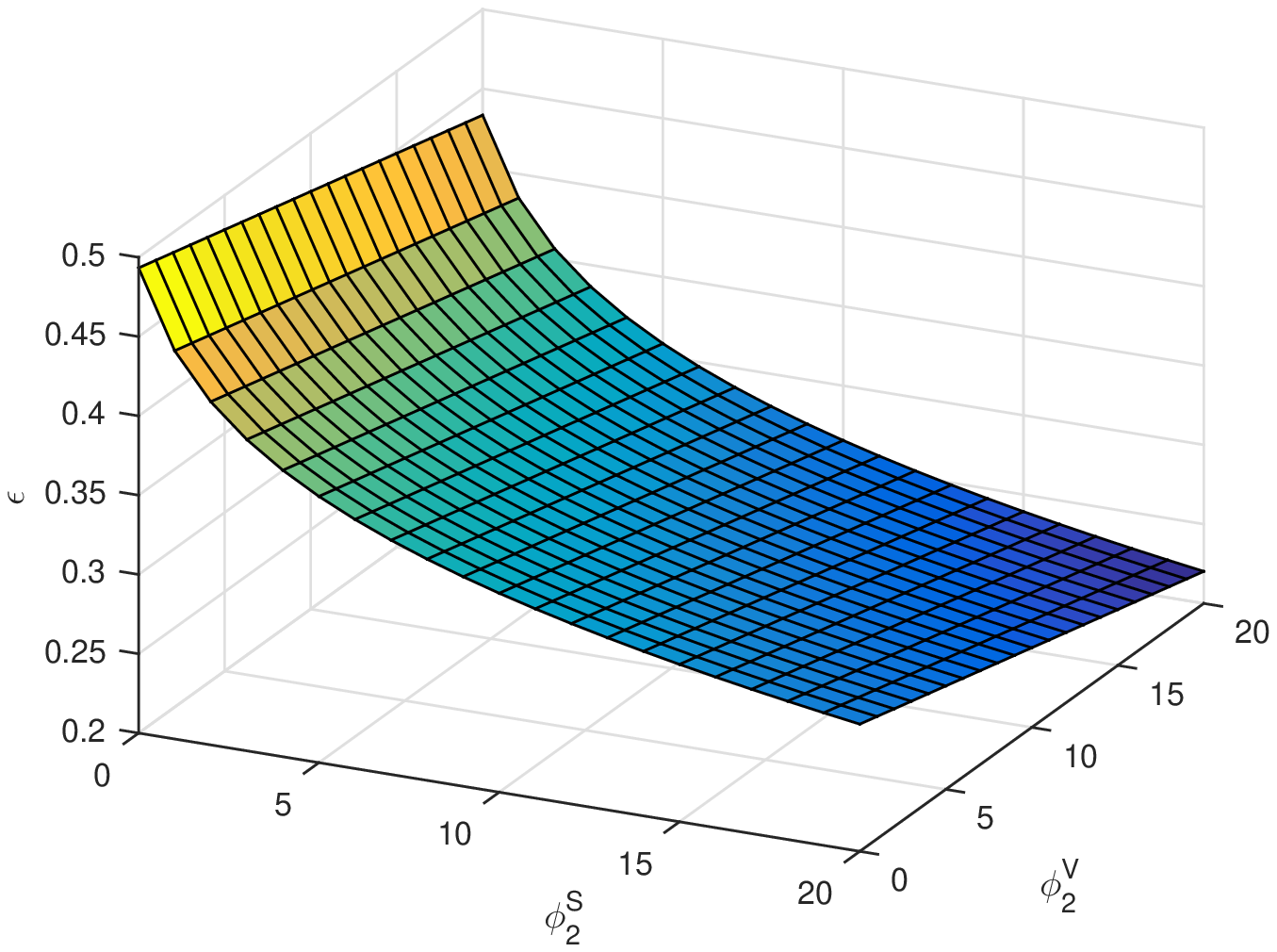}
}
\caption{The effects of $\phi_j^S$ and $\phi_j^V$ on the detection-error probability $\varepsilon$.}
\end{figure}

\begin{figure}
\centering
\subfigure[Adjustment to $\lambda_1$ ($e^S_2/\sqrt{V_2}$)]
{\label{fig2a}
\includegraphics[width=0.4\textwidth]{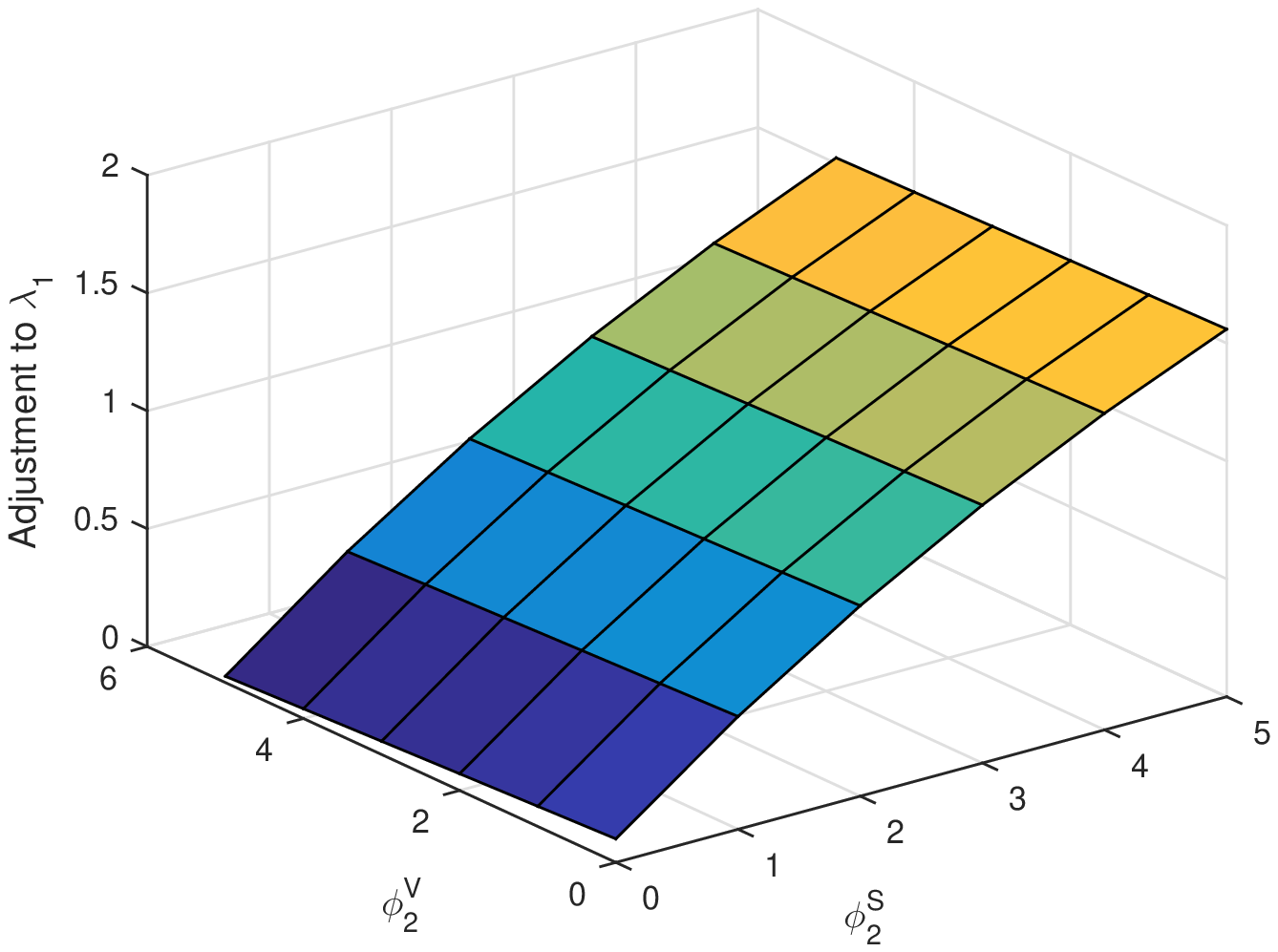}}
\subfigure[Adjustment to $\lambda_2$ ($e^V_2/\sqrt{V_2}$)]
{\label{fig2b}
\includegraphics[width=0.4\textwidth]{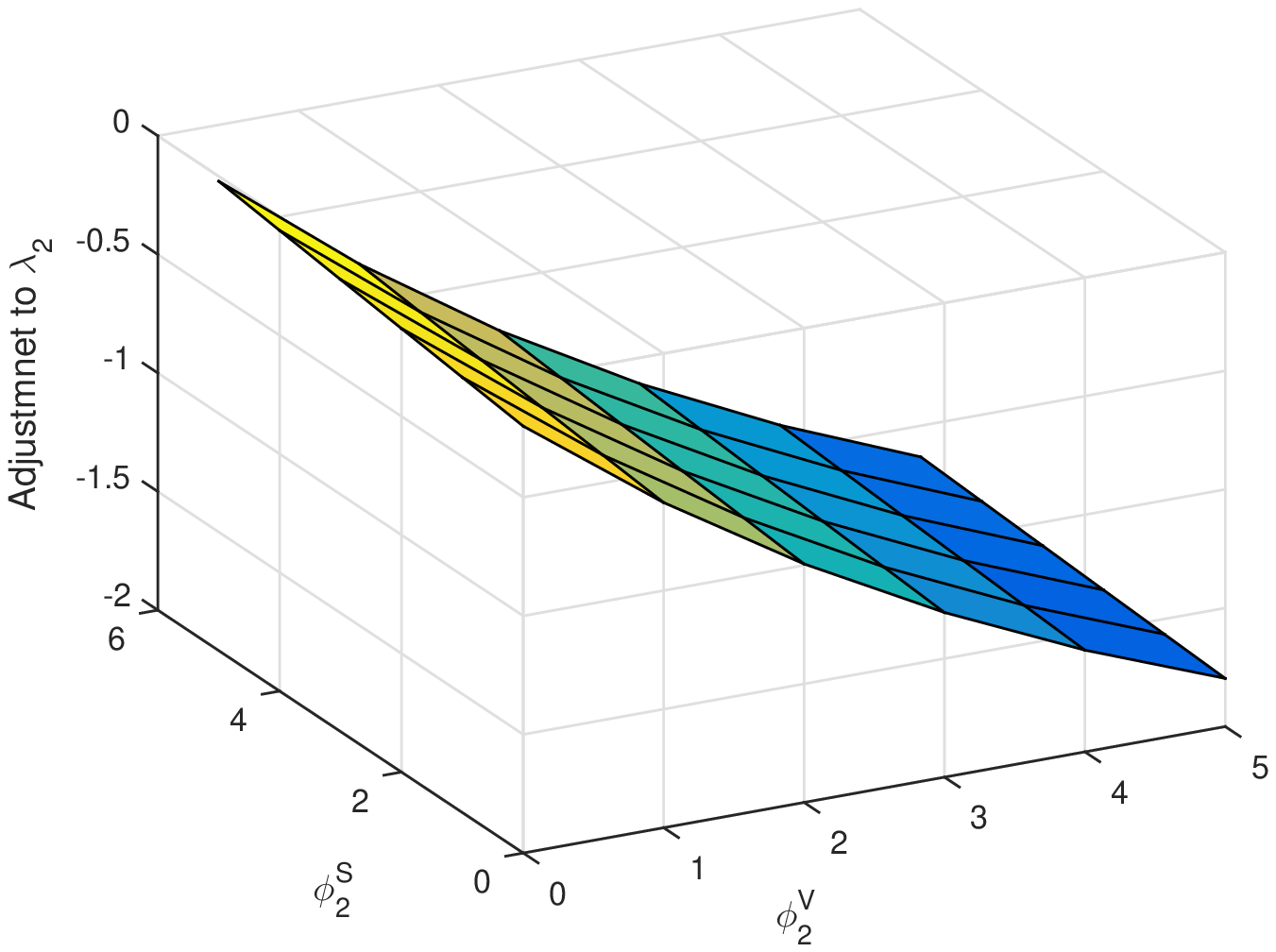}}
\caption{The effects of $\phi_2^S$ and $\phi_2^V$ on adjustments to $\lambda_1$ and $\lambda_2$ in complete markets.}
\label{fig2}
\end{figure}
Figures \ref{fig2a} and \ref{fig2b} show the adjustments $e^S_2/\sqrt{V_2}$ {  and $e^V_2/\sqrt{V_2}$   for second component uncertainty ($\phi_2^S$ and $\phi_2^V$) in a complete market. The adjustment $e^S_2/\sqrt{V_2}$ varies sharply with $\phi^S_2$, but only slightly with $\phi^V_2$, and vice versa for  $e^V_2/\sqrt{V_2}$. This result is consistent with that of } the single volatility model in \cite{Escobar15}.

\begin{figure}
\centering
\subfigure[The effects of $\phi_j^S$ and $\phi_j^V$ on $\beta^S_1$]
{\label{fig3a}
\includegraphics[width=0.4\textwidth]{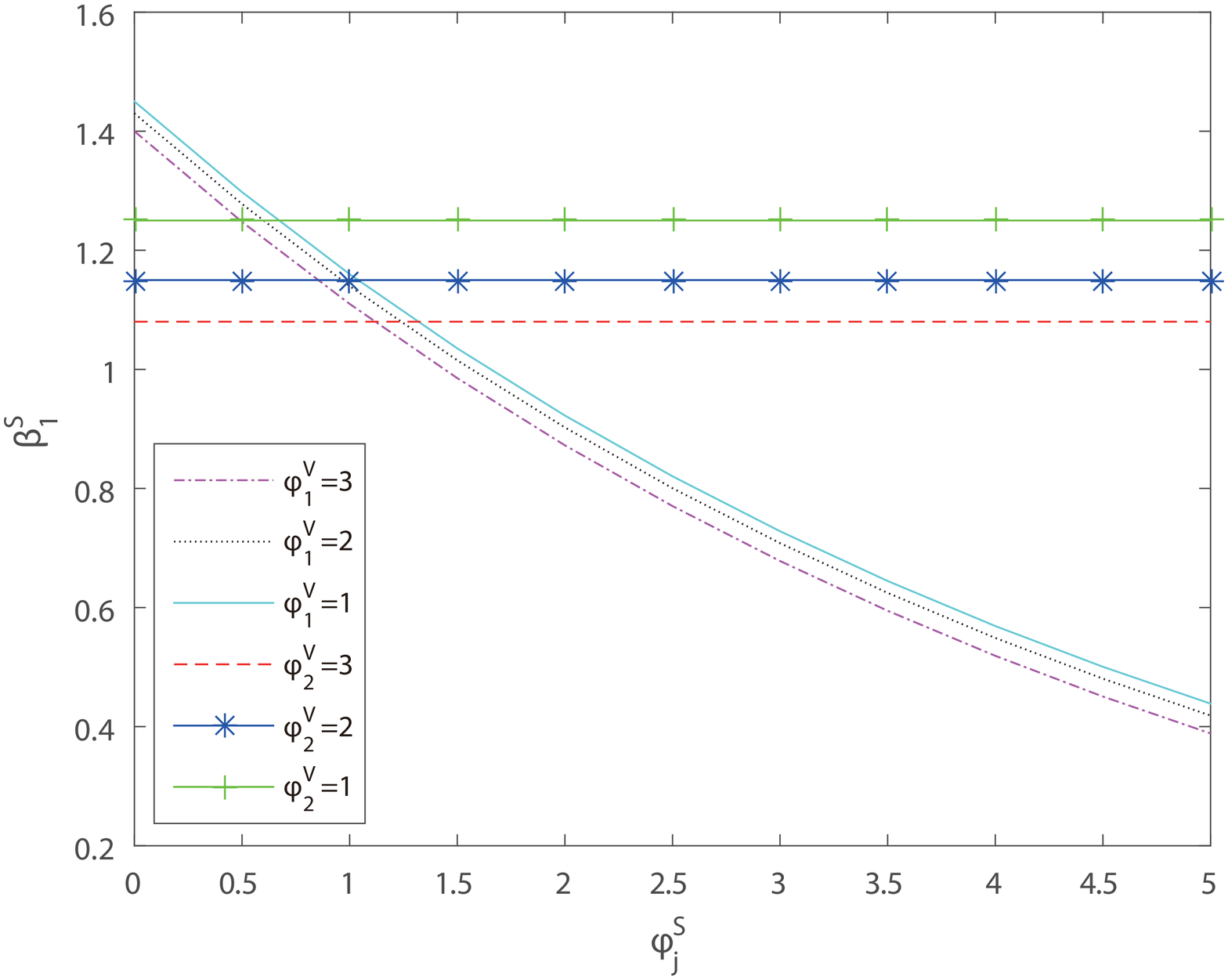}}
\subfigure[The effects of $\phi_j^S$ and $\phi_j^V$ on $\beta^S_2$]
{\label{fig3b}
\includegraphics[width=0.4\textwidth]{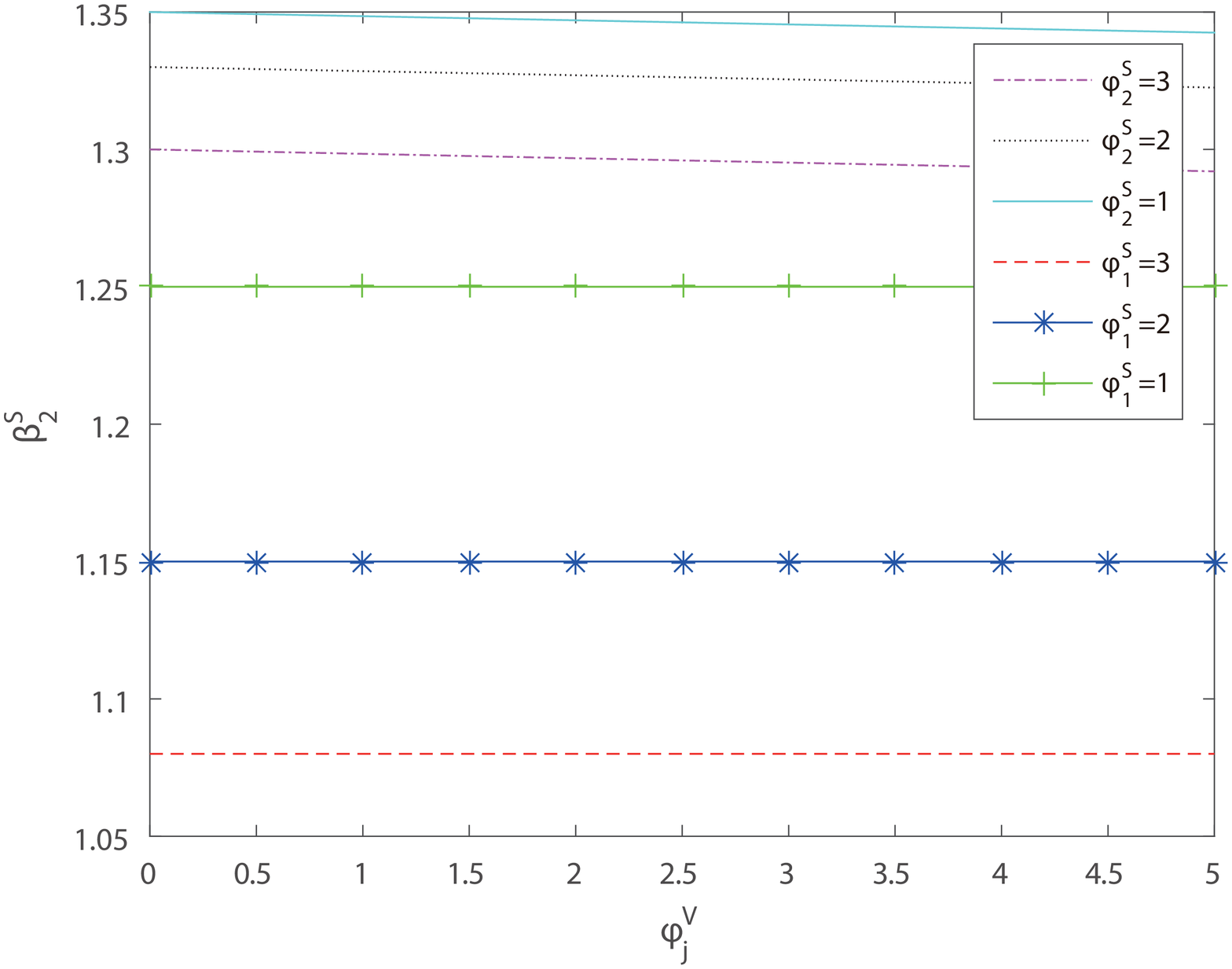}}
\subfigure[The effects of $\phi_j^S$ and $\phi_j^V$ on $\beta^V_1$]
{\label{fig3c}
\includegraphics[width=0.4\textwidth]{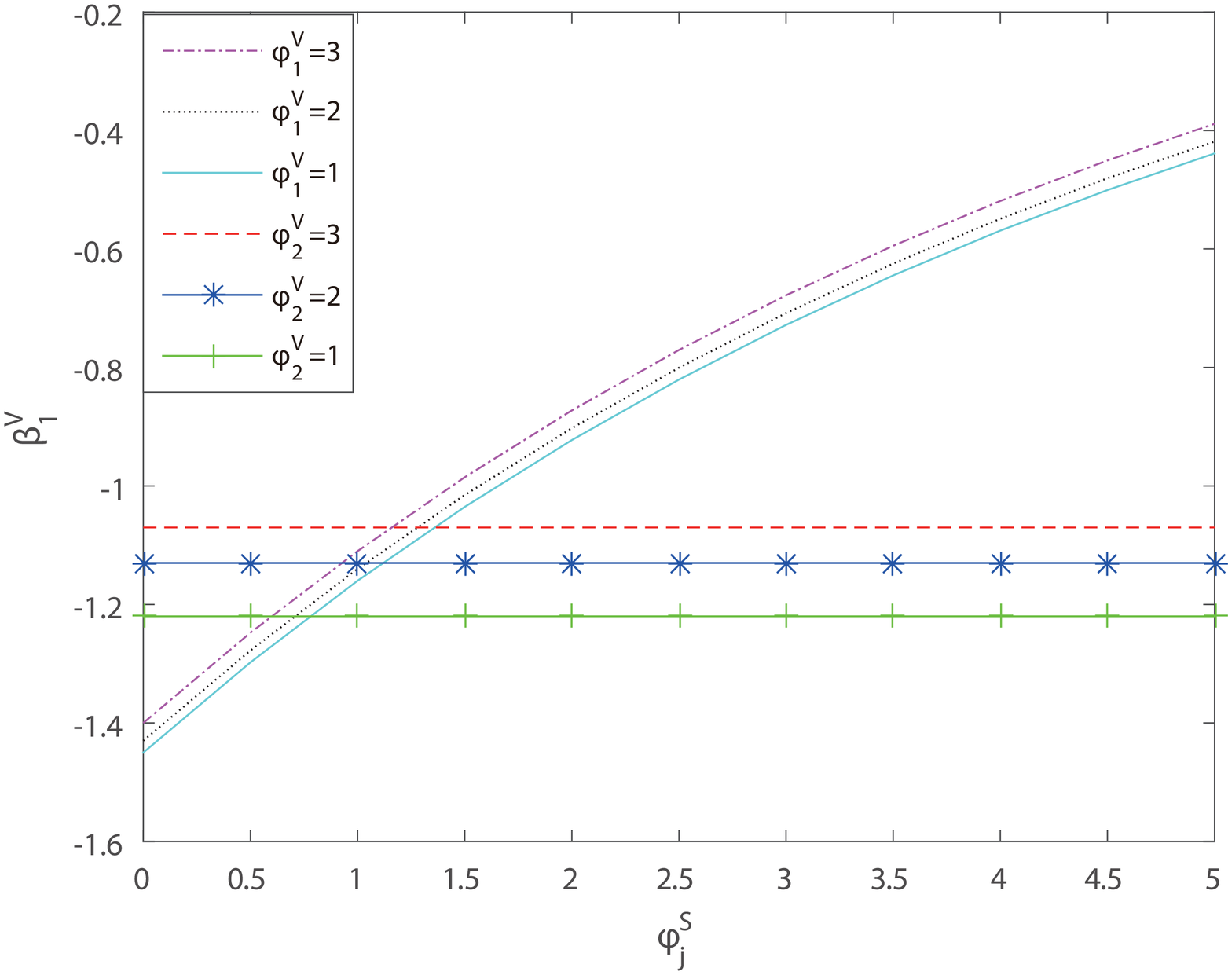}}
\subfigure[The effects of $\phi_j^S$ and $\phi_j^V$ on $\beta^V_2$]
{\label{fig3d}
\includegraphics[width=0.4\textwidth]{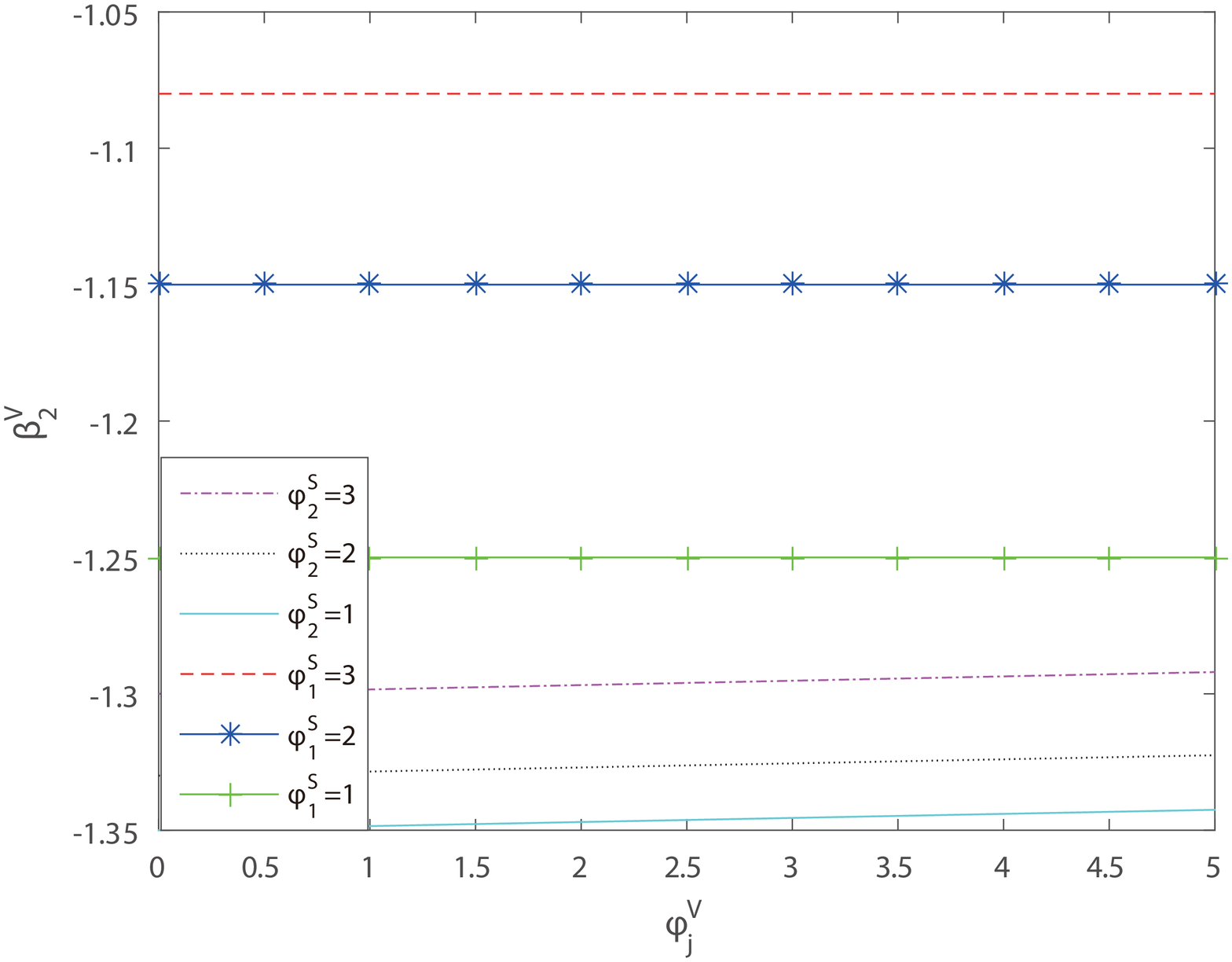}}
\caption{Optimal exposures in complete markets. }
\label{esa}
\end{figure}

\begin{figure}
\centering
\subfigure[The effects of $\phi_j^S$ and $\phi_j^V$ on $\beta^S_1$]
{\label{fig4a}
\includegraphics[width=0.43\textwidth]{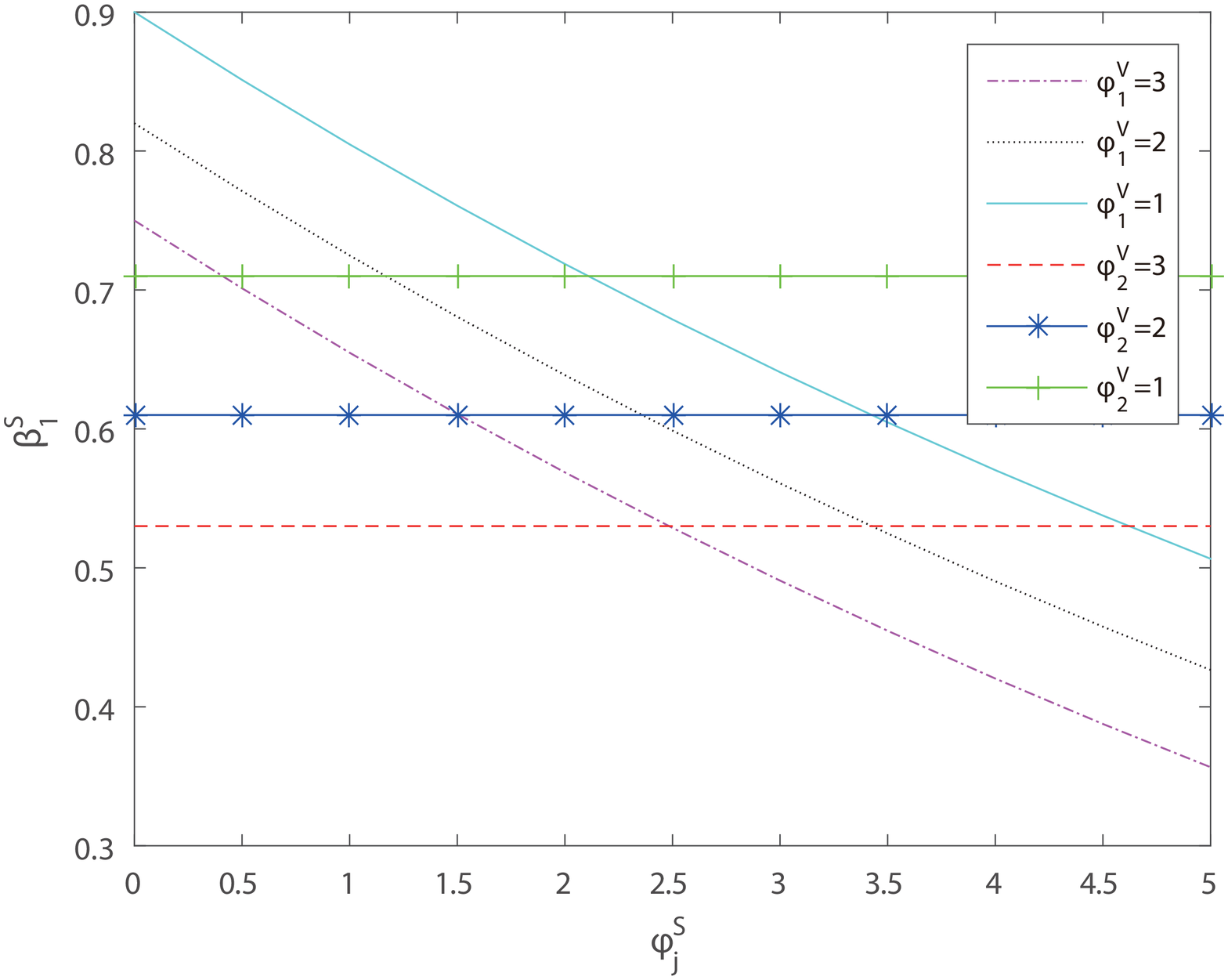}}
\subfigure[The effects of $\phi_j^S$ and $\phi_j^V$ on $\beta^S_2$]
{\label{fig4b}
\includegraphics[width=0.43\textwidth]{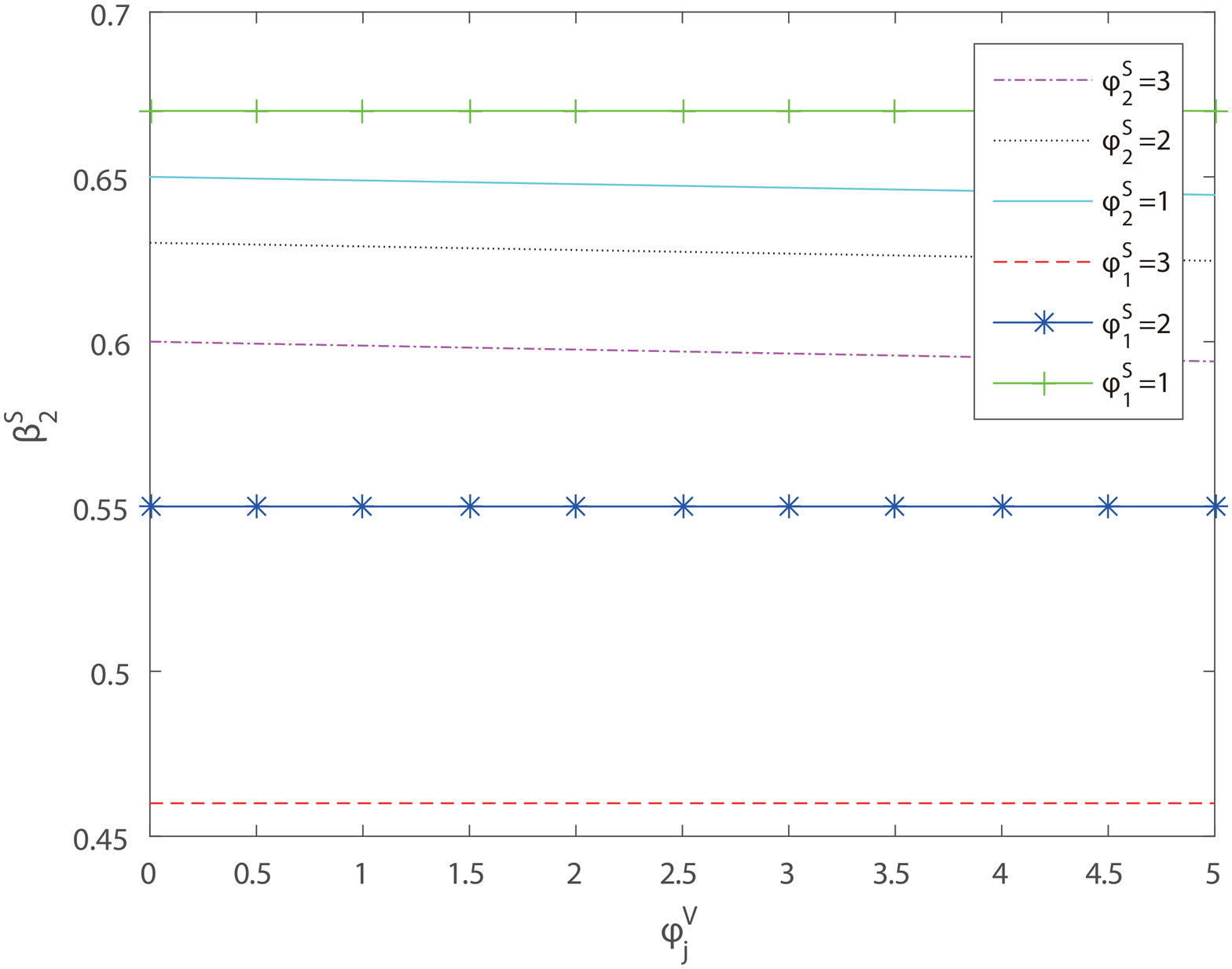}}
\caption{Optimal exposures in incomplete markets. }
\label{esb}
\end{figure}

Figures \ref{esa} and \ref{esb} { display} the optimal exposures in complete markets and incomplete markets, respectively. Firstly, stock risk exposure is more sensitive to ambiguity about the stock risk than the volatility risk and the volatility risk exposure is more sensitive to ambiguity about the volatility risk than the stock risk. Second, the ambiguity in one component of the volatility has no effect on the optimal exposure of the other component of the volatility. It should be pointed out that the optimal stock exposure decreases as the stock ambiguity parameter increases, whereas the optimal volatility exposure increases as the volatility ambiguity parameter increases. The differences between the results for the two markets are { significant}, indicating that derivative trading is of great importance for hedging risk and thus for making better portfolio choices.

\begin{figure}
\centering
\subfigure[The effects of $\phi_1^S$ and $\phi_1^V$
 in complete \newline \hfill  markets. ]
{\label{fig5a}
\includegraphics[width=0.4\textwidth]{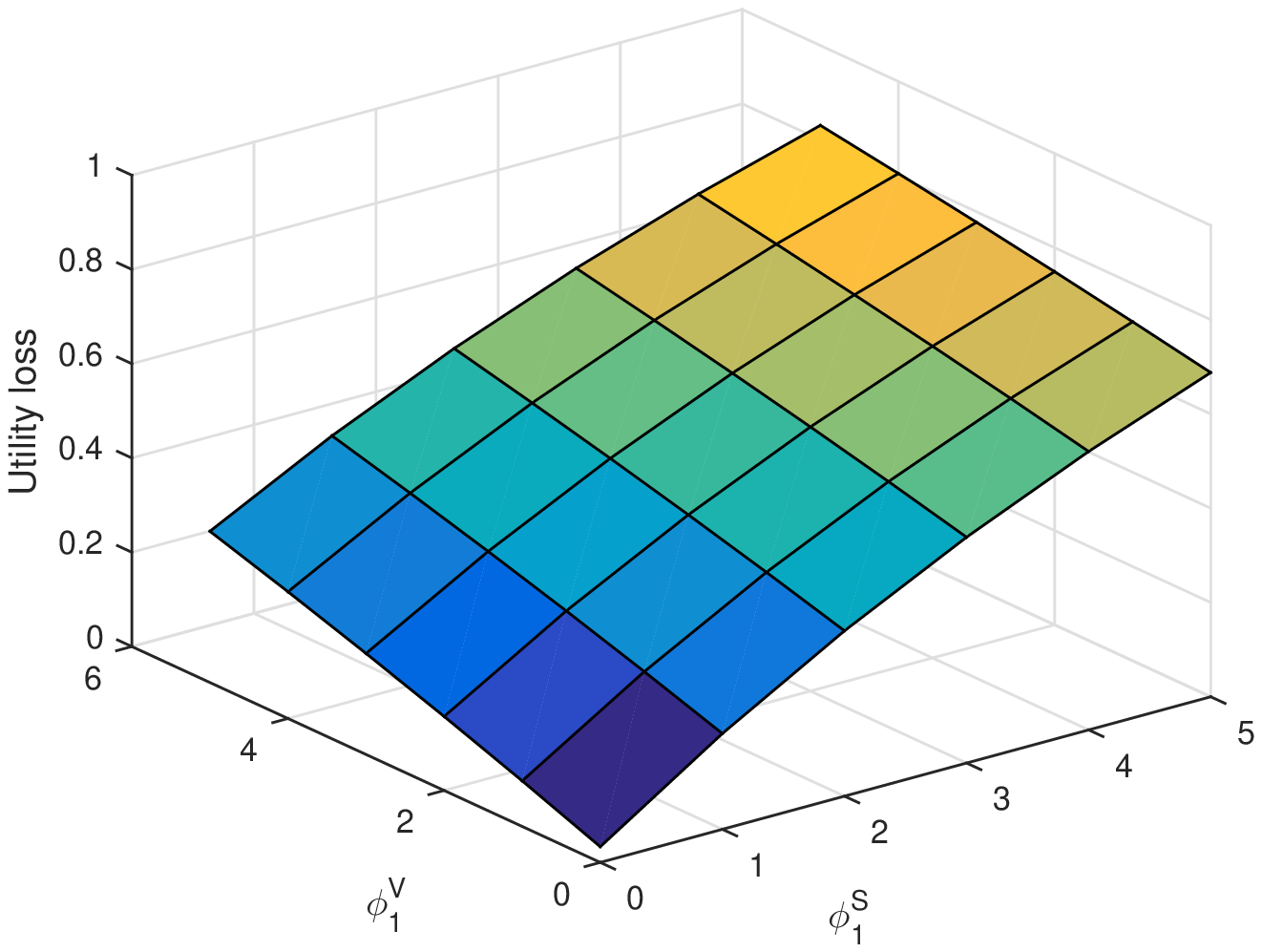}}
\subfigure[The effects of $\phi_2^S$ and $\phi_2^V$ in complete
\newline \hfill markets.]
{\label{fig5b}
\includegraphics[width=0.4\textwidth]{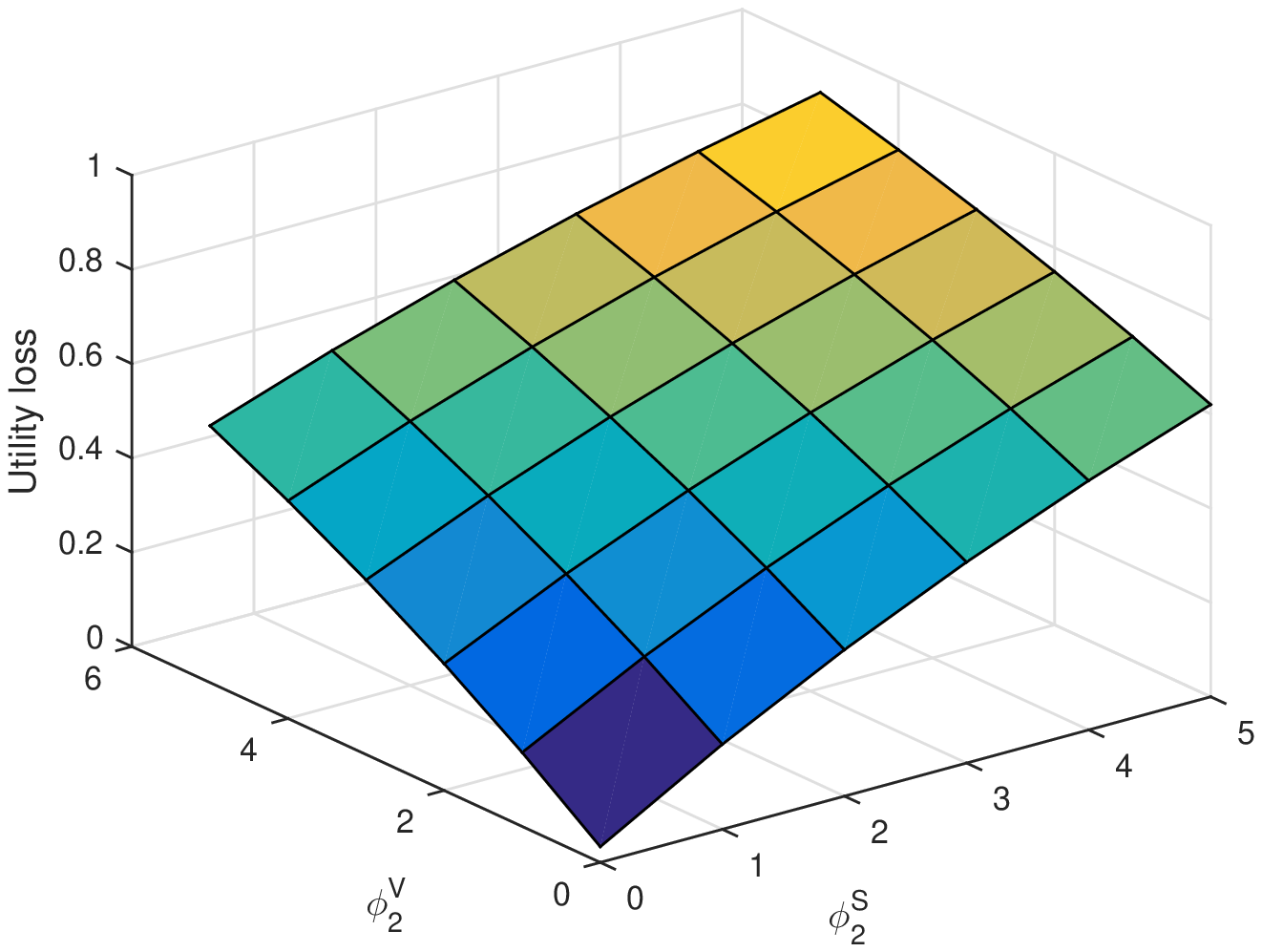}}
\subfigure[The effects of $\phi_1^S$ and $\phi_1^V$ in  incomplete
\newline \hfill markets.]
{\label{fig5c}
\includegraphics[width=0.4\textwidth]{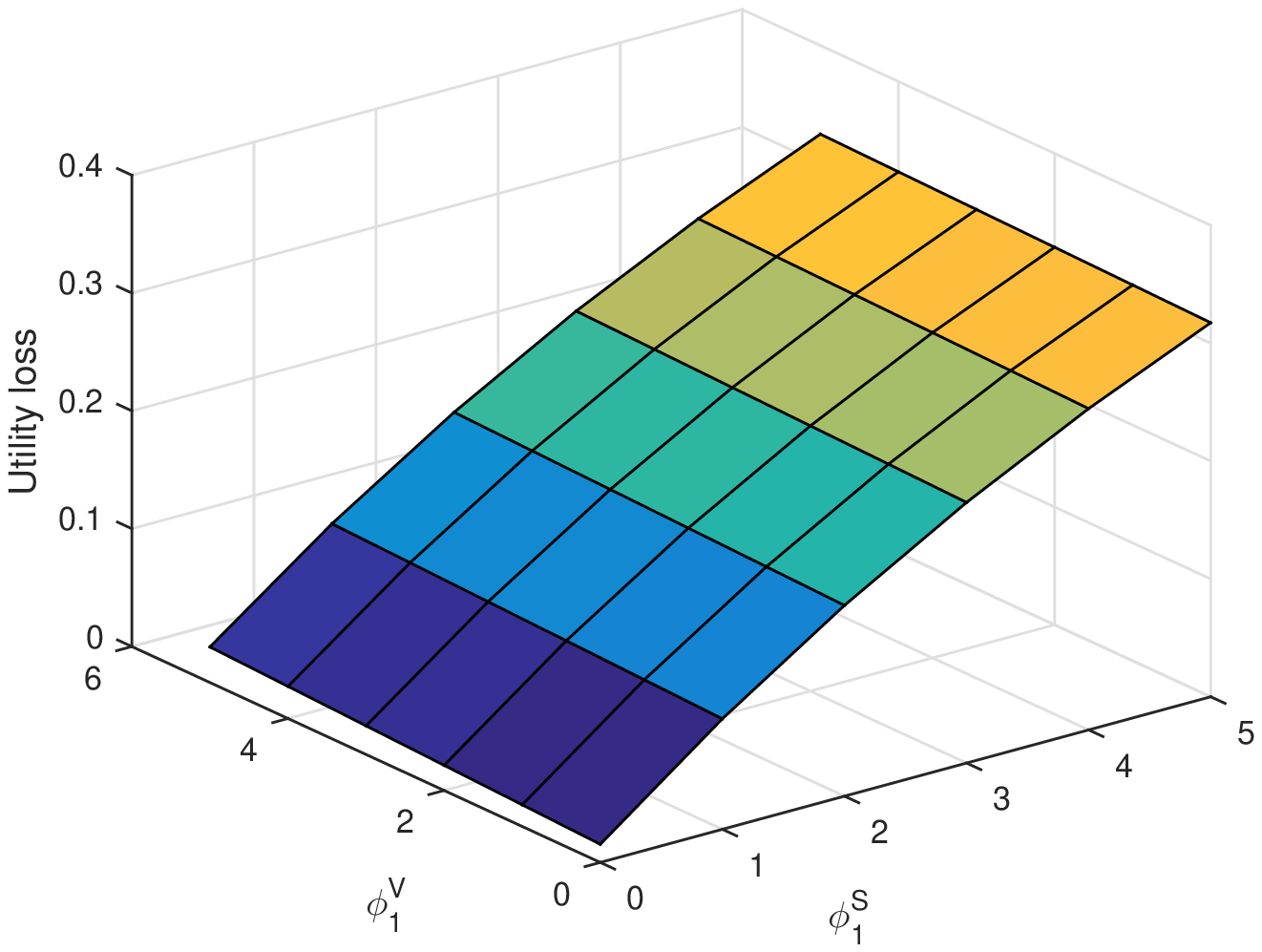}}
\subfigure[The effects of $\phi_2^S$ and $\phi_2^V$ incomplete
\newline \hfill markets.]
{\label{fig5d}
\includegraphics[width=0.4\textwidth]{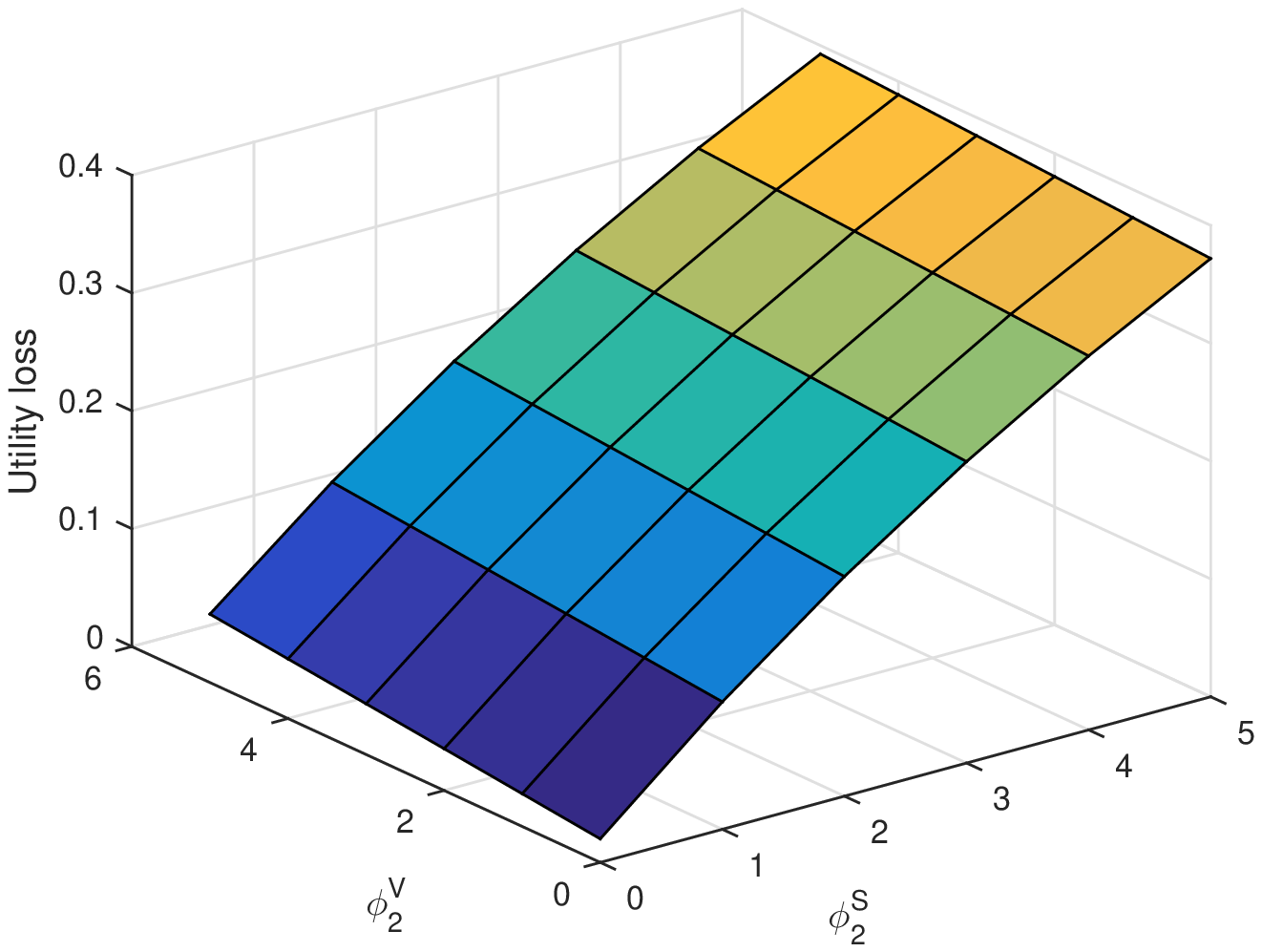}}
\caption{Utility loss from ignoring model uncertainty.}
\label{ula}
\end{figure}

Figure \ref{ula} { depicts} the variation of the utility (welfare) loss with ambiguity about the volatility and stock risks.  { The utility loss increases sharply as the ambiguity about the stock price increases, and} the increase is more significant in the complete market than that in the incomplete market. { The ambiguity about the volatility contributes to more utility loss in the complete market than} in the incomplete market. It is also observed that the two volatility factors affect the utility loss at different rates, thus one cannot neglect the contribution of the multi-factor volatility structure.

\section{Conclusions}

 In this work, we study the robust optimal portfolio selection problem for a risk-averse and ambiguity-averse investor, who  invests in a money market account, a stock, whose price follows the multi-factor volatility model, and several correlated derivatives. Our results indicate that the ambiguity about asset price and volatility can influence investment decisions in both complete markets and incomplete markets, and therefore one should always search for a robust optimal portfolio.
Jump risks and correlated volatilities have significant effect in robust optimal portfolio selection problems. Thus, in the presence of jump risks an investor should consider additional derivatives to hedge the associated risks, and when choosing a robust portfolio, one needs to consider the worst-case measure due to the jump risks and determine the optimal strategy accordingly.

Although our new analytical approximate solution contributes to the literature of robust optimal portfolio selection by providing evidences that the correlated volatility-factors have impact on the worst-case scenario measures and robust optimal strategies, more research is needed in this area. For example, future work could include investigating the case when there is a time-varying correlation instead of a constant correlation coefficient in our current model. Further, it is of interest to find out whether or not including information asymmetry or trading obstacles (transaction costs or liquidity) leads to better models.


{\small
\begin{appendices}
\singlespacing
\section{{ Proof of} Proposition \ref{prop1}}\label{A1}
Solving optimal problem \eqref{HJB0} with respect to $e^S_j$ and $e_j^V$,
we have
\begin{equation}\label{es0}
(e^S_j)^*=\Psi_j^S(x\beta^S_jJ_x+\rho_j\sigma_jJ_{v_j})\sqrt{v_j},\quad
(e^V_j)^*=\Psi_j^V(x\beta^V_jJ_x+\sqrt{1-\rho_j^2}\sigma_jJ_{v_j})\sqrt{v_j}.
\end{equation}
Substituting \eqref{es0} into \eqref{HJB0}, we obtain the following equation
\begin{equation}\label{HJBA1}
\begin{aligned}&\sup_{\beta^S_i,\beta^V_i}
\bigg\{J_t+x\big(r+\sum_{j=1}^{2}(\beta^{S}_j\lambda_jv_j
+\beta^{V}_j\mu_jv_j\big)J_x
+\frac{1}{2}x^2\sum_{j=1}^{2}[(\beta^{S}_j)^2+(\beta^{V}_j)^2]v_jJ_{xx}
+\sum_{j=1}^{2}\big(\kappa_j(\theta_j-v_j)\big)J_{v_j}\\
&+\frac{1}{2}\sum_{j=1}^{2}\sigma_j^2v_jJ_{v_jv_j}+\sum_{j=1}^{2}\sigma_jv_jx\big(\beta^S_j\rho_j+\beta^V_j\sqrt{1-\rho_j^2}\big)J_{v_jx}-\sum_{j=1}^{2}\frac{1}{2}\Psi_j^Sv_j\left[x^2(\beta^S_j)^2J_x^2+2x\beta^S_j\rho_j\sigma_jJ_xJ_{V_j}+\rho_j^2\sigma_j^2J^2_{v_j}\right]\\
&-\sum_{j=1}^{2}\frac{1}{2}\Psi_j^Sv_j\left[x^2(\beta^V_j)^2J_x^2+2x\beta^V_j\sqrt{1-\rho_j^2}\sigma_jJ_xJ_{v_j}+(1-\rho_j^2)\sigma_j^2J^2_{v_j}\right]
\bigg\}=0.
\end{aligned}
\end{equation}
Assuming that the solution $J$ is of the form
$J(t,x,v_1,v_2)=\frac{x^{1-\gamma}}{1-\gamma}\exp(H_1(T-t)v_1+H_2(T-t)v_2+h(T-t))$
and choosing
$$\Psi^S_j=\frac{\phi_j^S}{(1-\gamma)J(t,x,v_1,v_2)}, \quad \Psi^V_j=\frac{\phi_j^V}{(1-\gamma)J(t,x,v_1,v_2)}, \quad \phi_j^S, \ \phi_j^V>0,$$
Equation \eqref{HJBA1} yields
\begin{equation}\label{HJBA20}
\begin{aligned}
\sup_{\beta^S_i,\beta^V_i}
\bigg\{
&-h'-\sum_{j=1}^{2}(H_j)'v_j+(1-\gamma)\big(r+\sum_{j=1}^{2}(\beta^{S}_j\lambda_jv_j
+\beta^{v}_j\mu_jv_j\big)\big)
-\frac{\gamma(1-\gamma)}{2}\sum_{j=1}^{2}[(\beta^{S}_j)^2+(\beta^{v}_j)^2]v_j\\
&+\sum_{j=1}^{2}\big(\kappa_j(\theta_j-v_j))H_j
+\frac{1}{2}\sum_{j=1}^{2}\sigma_j^2v_j(H_j)^2+\sum_{j=1}^{2}\sigma_jv_jx\big(\beta^S_j\rho_j+\beta^V_j\sqrt{1-\rho_j^2}\big)H_j\\
&-\sum_{j=1}^{2}\frac{1}{2}(1-\gamma)\phi_j^Sv_j\left[(\beta^S_j)^2+2\beta^S_j\rho_j\sigma_j\frac{H_j}{1-\gamma}+\rho_j^2\sigma_j^2\frac{(H_j)^2}{(1-\gamma)^2}\right]\\
&-\sum_{j=1}^{2}\frac{1}{2}(1-\gamma)\phi_j^Vv_j\left[(\beta^V_j)^2+2\beta^V_j\sqrt{1-\rho_j^2}\sigma_j\frac{H_j}{1-\gamma}+(1-\rho_j^2)\sigma_j^2\frac{(H_j)^2}{(1-\gamma)^2}\right]
\bigg\}=0,
\end{aligned}
\end{equation}
where $h'$ represents the time derivative of function $h(t)$. From \eqref{HJBA20}, we derive the optimal exposures as
\begin{equation}\label{ope}
(\beta^S_j)^*=\frac{\lambda_j}{\gamma+\phi^S_j}+\frac{(1-\gamma-\phi^S_j)\sigma_j\rho_j}{(1-\gamma)(\gamma+\phi^S_j)}H_j(T-t),\quad (\beta^V_j)^*=\frac{\mu_j}{\gamma+\phi^V_j}+\frac{(1-\gamma-\phi^V_j)\sigma_j\sqrt{1-\rho_j^2}}{(1-\gamma)(\gamma+\phi^V_j)}H_j(T-t).
\end{equation}
To determine $H_1$ and $H_2$ and $h$, we substitute \eqref{ope} into \eqref{HJBA20} to obtain
\begin{equation}\label{HJBA3}
\begin{aligned}
&-h'-\sum_{j=1}^{2}(H_j)'v_j+(1-\gamma)(
r+\sum_{j=1}^{2}\big(\kappa_j(\theta_j-v_j))H_j)
+\frac{1}{2}\sum_{j=1}^{2}\sigma_j^2v_j(H_j)^2-\sum_{j=1}^{2}\frac{1}{2}\phi_j^Sv_j\rho_j^2\sigma_j^2\frac{(H_j)^2}{1-\gamma}\\
&-\sum_{j=1}^{2}\frac{1}{2}\phi_j^Vv_j(1-\rho_j^2)\sigma_j^2\frac{(H_j)^2}{1-\gamma}+\sum_{j=1}^{2}\left[\frac{(1-\gamma)\lambda_1^2}{2(\gamma+\phi_j^S)}+\frac{\lambda_1(1-\gamma-\phi_j^S)\sigma_j\rho_iH_j}{\gamma+\phi^S_j}
+\frac{(1-\gamma-\phi_i^S)^2\sigma^2_i\rho^2_i(H_j)^2}{2(1-\gamma)(\gamma+\phi^S_i)}
\right]v_j\\
&+\sum_{j=1}^{2}\left[\frac{(1-\gamma)\lambda_2^2}{2(\gamma+\phi_1^V)}+\frac{\lambda_2(1-\gamma-\phi_j^V)\sigma_j\sqrt{1-\rho^2_i}H_j}{\gamma+\phi^S_j}
+\frac{(1-\gamma-\phi_j^V)^2\sigma^2_j(1-\rho^2_j)(H_j)^2}{2(1-\gamma)(\gamma+\phi^V_j)}
\right]v_j=0.
\end{aligned}
\end{equation}
 Comparing the terms concerned with $v_j$, we get the following Riccati equations
\begin{equation}\label{e24}
\left\{
\begin{aligned}
&H_j'=a_jH_j+b_j(H_j)^2+c_j,\quad  H_j(0)=0,\ j=1, 2, \\
&h'=\kappa_1\theta_1H_1+\kappa_2\theta_2H_2+(1-\gamma)r,\quad h(0)=0.
\end{aligned}
\right.
\end{equation}
Let $d_j=\sqrt{a_j^2-4b_jc_j}$\, , where
\begin{eqnarray*}
a_j&=&-\kappa_j+\frac{\lambda_1(1-\gamma-\phi_j^S)\sigma_j\rho_j}{\gamma+\phi_j^S}+\frac{\lambda_2(1-\gamma-\phi_j^V)\sigma_j}{\gamma+\phi_j^V}\sqrt{1-\rho_j^2},\\
b_j&=&\frac{\sigma_j^2}{2}-\frac{\phi^S_j\rho_j^2\sigma_j^2}{2(1-\gamma)}-\frac{\phi^V_j(1-\rho_j^2)\sigma_j^2}{2(1-\gamma)}+\frac{(1-\gamma-\phi_j^S)\sigma^2_j\rho^2_j}{2(1-\gamma)(\gamma+\phi_j^S)}+\frac{(1-\gamma-\phi_j^V)^2\sigma^2_j(1-\rho_j^2)}{2(1-\gamma)(\gamma+\phi_j^V)},\\
c_j&=&\frac{(1-\gamma)\lambda_1^2}{2(\gamma+\phi^S_j)}+\frac{(1-\gamma)\lambda_2^2}{2(\gamma+\phi^V_j)}.
\end{eqnarray*}
Functions $H_j$ and $h$  can be obtained from \eqref{e24}, as shown in \eqref{Hh1}. 
\qed
\begin{section}{{Proof of} Proposition \ref{propa2}}\label{A2}
\begin{proof}
It is sufficient to show that the Novikov's condition is satisfied for the worst-case probability measure, that is,
$$\mathrm{E}^{\mathbb{P}}\left[\exp\left\{\frac{1}{2}\int_{0}^{T}((e^S_1(t))^*)^2+((e^S_2(t))^*)^2+((e^V_1(t))^*)^2+((e^V_2(t))^*)^2dt\right\}\right]<\infty.$$
Applying \eqref{opte} in Proposition \ref{prop1}, one obtains
$\mathrm{E}^{\mathbb{P}}\left[\exp\left\{\frac{1}{2}\int_{0}^{T}\left(K_1(T-t)V_1(t)+K_2(T-t)V_2(t)\right)dt\right\}\right]<\infty,$
where
$$K_j(T-t)=\left(\frac{\lambda_j}{\gamma+\phi^S_j}+\frac{\sigma_j\rho_j H_j(T-t)}{(1-\gamma)(\gamma+\phi^S_j)}\right)^2(\phi^S_j)^2
+\left(\frac{\mu_j}{\gamma+\phi^V_j}+\frac{\sigma_j\sqrt{1-\rho_j^2} H_j(T-t)}{(1-\gamma)(\gamma+\phi^V_j)}\right)^2(\phi^V_j)^2,\ j= 1, 2. $$
Define $k_1=\sup_{t\in[0,T]}K_1(T-t)$ and $k_2=\sup_{t\in[0,T]}K_2(T-t)$, then
\begin{equation}\label{ieq}
\begin{aligned}
\mathrm{E}^{\mathbb{P}}\left[\exp\left\{\frac{1}{2}\int_{0}^{T}\left(K_1(T-t)V_1(t)+K_2(T-t)V_2(t)\right)dt\right\}\right]
<
\sum\limits_{i=1}^2\mathrm{E}^{\mathbb{P}}\left[\exp\left\{\frac{k_i}{2}\int_{0}^{T}V_i(t)dt\right\}\right] 
<\infty.
\end{aligned}
\end{equation}
The first inequity in \eqref{ieq} holds because $V_1$ and $V_2$ are independent, while the second holds if $k_1\leq \frac{\kappa_1^2}{\sigma_1^2}$
and $k_2\leq \frac{\kappa_2^2}{\sigma_2^2}$ (see \cite{Kraft05}). Recall that $\gamma>1$, $\lambda_j > 0$, $\rho_j<0$ and $\mu_j<0$. Thus,  $K_j$ reaches its maximum value at $t=T$ because $H_j$ is maximum at $t= T$.
Therefore,
$$k_1=(\phi_1^S)^2\frac{\lambda_1^2}{(\gamma+\phi^S_1)^2}+(\phi_1^V)^2\frac{\mu^2_1}{(\gamma+\phi^V_1)^2}, \quad k_2=(\phi_2^S)^2\frac{\lambda^2_2}{(\gamma+\phi^S_2)^2}+(\phi_2^V)^2\frac{\mu^2_2}{(\gamma+\phi^V_2)^2},$$
which yield the desired results immediately.
\qed
\end{proof}
\section{{Proof }of Proposition \ref{prop2}}\label{A3}
\begin{proof}
In the incomplete market, solving the optimal problem \eqref{wealth2}-\eqref{HJB02}, we obtain the general optimal wealth invested in stock
\begin{equation}\label{pis2}
\pi^S=\frac{\sum_{j=1}^{2}\left[(1-\gamma)v_j(\lambda_j+\rho_j\sigma_j\bar{H}_j)-\phi_j^S\rho_j\sigma_jv_j\bar{H}_j\right]}{\sum_{j=1}^{2}(1-\gamma)v_j(\gamma+\phi_j^S)}.
\end{equation}
Substituting \eqref{pis2} into \eqref{HJB02},
we obtain the following general HJB equation
\begin{equation}\label{HJBA04}
\begin{aligned}
&\bar{h}'+\sum_{j=1}^{2}(\bar{H}_j)'v_j-(1-\gamma)\big(r+\sum_{j=1}^{2}\pi^S\lambda_jv_j
\big)
+\frac{\gamma(1-\gamma)}{2}\sum_{j=1}^{2}(\pi^S)^2v_j+\sum_{j=1}^{2}\big(\kappa_j(\theta_j-v_j))\bar{H}_j
\\
&-\frac{1}{2}\sum_{j=1}^{2}\sigma_j^2v_j(\bar{H}_j)^2-\sum_{j=1}^{2}\sigma_jv_jx\rho_j\bar{H}_j\pi^S+\sum_{j=1}^{2}\frac{1}{2}(1-\gamma)\phi_j^Sv_j\left[(\pi^S)^2+2\pi^S(\bar{H}_j)\rho_j\sigma_j\frac{\bar{H}_j}{1-\gamma}+\rho_j^2\sigma_j^2\frac{(\bar{H}_j)^2}{(1-\gamma)^2}\right]\\
&+\sum_{j=1}^{2}\frac{1}{2}(1-\gamma)\phi_j^Vv_j(1-\rho_j^2)\sigma_j^2\frac{(\bar{H}_j)^2}{(1-\gamma)^2}=0.
\end{aligned}
\end{equation}
If there only exists one risk factor, say,  $W_1$, from \eqref{pis2} we have
$$\pi^S=\frac{\lambda_1}{\gamma+\phi^S_1}+\frac{(1-\gamma-\phi^S_1)\sigma_1\rho_1}{(1-\gamma)(\gamma+\phi^S_1)}\bar{H}_1(\tau).$$
Alternatively,  if the only risk factor is $W_2$,
$\pi^S=\frac{\lambda_2}{\gamma+\phi^S_2}+\frac{(1-\gamma-\phi^S_2)\sigma_2\rho_2}{(1-\gamma)(\gamma+\phi^S_2)}\bar{H}_2(\tau).$

Further, if the risk factors are the same, $W_1=W_2$, then
$$\pi^S=\frac{\lambda_1}{\gamma+\phi^S_1}+\frac{(1-\gamma-\phi^S_1)\sigma_1\rho_1}{(1-\gamma)(\gamma+\phi^S_1)}\bar{H}_1(\tau)=\frac{\lambda_2}{\gamma+\phi^S_2}+\frac{(1-\gamma-\phi^S_2)\sigma_2\rho_2}{(1-\gamma)(\gamma+\phi^S_2)}\bar{H}_2(\tau).$$
Therefore, under the circumstances that there are only one single risk or two identical risk factors we can write,  without loss of generality,
\begin{equation*}
\begin{aligned}
\pi^S=\pi_j(\bar{H}_j) = \frac{\lambda_j}{\gamma+\phi^S_j}+\frac{(1-\gamma-\phi^S_j)\sigma_j\rho_j}{(1-\gamma)(\gamma+\phi^S_j)}\bar{H}_j(\tau),\ j = 1, 2.
\end{aligned}
\end{equation*}
Substituting $\pi^S=\pi_j(\bar{H}_j)$ into \eqref{HJB02}, we obtain the following HJB equation
\begin{equation}\label{HJBA4}
\begin{aligned}
0=&\bar{h}'+\sum_{j=1}^{2}(\bar{H}_j)'v_j-(1-\gamma)\big(r+\sum_{j=1}^{2}\pi_j(\bar{H}_j)\lambda_jv_j
\big)
+\frac{\gamma(1-\gamma)}{2}\sum_{j=1}^{2}\pi^2_j(\bar{H}_j)v_j+\sum_{j=1}^{2}\big(\kappa_j(\theta_j-v_j))\bar{H}_j\\
&
-\frac{1}{2}\sum_{j=1}^{2}\sigma_j^2v_j(\bar{H}_j)^2-\sum_{j=1}^{2}\sigma_jv_jx\rho_j\bar{H}_j\pi_j(\bar{H}_j)+\sum_{j=1}^{2}\frac{1}{2}(1-\gamma)\phi_j^Vv_j(1-\rho_j^2)\sigma_j^2\frac{(\bar{H}_j)^2}{(1-\gamma)^2}\\
&+\sum_{j=1}^{2}\frac{1}{2}(1-\gamma)\phi_j^Sv_j\left[(\pi_j(\bar{H}_j))^2+2\pi_j(\bar{H}_j)\rho_j\sigma_j\frac{\bar{H}_j}{1-\gamma}+\rho_j^2\sigma_j^2\frac{(\bar{H}_j)^2}{(1-\gamma)^2}\right].
\end{aligned}
\end{equation}
By matching coefficients, we obtain the ordinary differential equations for the solution of $\bar{H}_j$ and $\bar{h}$ as follows
\begin{equation}\label{HH2}
\left\{
\begin{aligned}
\bar{H}'_j&=(1-\gamma)\pi_j(\bar{H}_j)\lambda_j
-\frac{\gamma(1-\gamma)}{2}\pi^2_j(\bar{H}_j)-\kappa_j\bar{H}_j
+\frac{1}{2}\sigma_j^2(\bar{H}_j)^2+(1-\gamma)\sigma_j\rho_j\bar{H}_j\pi_j(\bar{H}_j)\\
&-\frac{1}{2}(1-\gamma)\phi_j^S\left[(\pi_j(\bar{H}_j))^2+2\pi_j(\bar{H}_j)\rho_j\sigma_j\frac{\bar{H}_j}{1-\gamma}+\rho_j^2\sigma_j^2\frac{(\bar{H}_j)^2}{(1-\gamma)^2}\right]\\
&-\frac{1}{2}(1-\gamma)\phi_j^V(1-\rho_j^2)\sigma_j^2\frac{(\bar{H}_j)^2}{(1-\gamma)^2},\quad && \bar{H}_j(0)=0,\\
\bar{h}'&=\kappa_1\theta_1\bar{H}_1+\kappa_2\theta_2\bar{H}_2+(1-\gamma)r,\quad && \bar{h}(0)=0.
\end{aligned}
\right.
\end{equation}
The worst-case measures can be computed directly by \eqref{e+} and thus the proof is complete.\qed
\end{proof}
\section{{Proof of} Proposition \ref{prop3}}\label{A4}
\begin{proof}
Without loss of generality, we assume the indirect utility function of all suboptimal strategy $\Pi$ as
\begin{equation}\label{e34}
J^{\Pi}=\frac{x^{1-\gamma}}{1-\gamma}\exp\left(H^{\Pi}_1(\tau)v_1+H^{\Pi}_2(\tau)v_2+h^{\Pi}(\tau)\right).
\end{equation}
For mathematical tractability, we choose
$$\displaystyle \Psi^S_j=\frac{\widetilde{\phi}_j^S}{(1-\gamma)J(t,x,v_1,v_2)} \quad \mbox{and}
\quad \Psi^V_j=\frac{\widetilde{\phi}_j^V}{(1-\gamma)J(t,x,v_1,v_2)},$$
where the ambiguity aversion parameters
$\widetilde{\phi}_j^S$, $\widetilde{\phi}_j^V>0$.

The functions $H^{\Pi}_1$, $H^{\Pi}_2$ and $h^{\Pi}$ in \eqref{e34} should satisfy the HJB PDE \textcolor{blue}{\eqref{HJB22}},  which leads to
\begin{equation}\label{HJBA2}
\begin{aligned}
0=&-(h^{\Pi})'-\sum_{j=1}^{2}(H^{\Pi}_j)'v_j+(1-\gamma)\big(r+\sum_{j=1}^{2}(\beta^{S}_j\lambda_jv_j
+\beta^{v}_j\mu_jv_j\big)\big)
-\frac{\gamma(1-\gamma)}{2}\sum_{j=1}^{2}[(\beta^{S}_j)^2+(\beta^{v}_j)^2]v_j\\
&+\sum_{j=1}^{2}\big(\kappa_j(\theta_j-v_j))H^{\Pi}_j
+\frac{1}{2}\sum_{j=1}^{2}\sigma_j^2v_j(H^{\Pi}_j)^2+\sum_{j=1}^{2}\sigma_jv_jx\big(\beta^S_j\rho_j+\beta^V_j\sqrt{1-\rho_j^2}\big)H^{\Pi}_j\\
&-\sum_{j=1}^{2}\frac{1}{2}(1-\gamma)\widetilde{\phi}_j^Sv_j\left[(\beta^S_j)^2+2\beta^S_j\rho_j\sigma_j\frac{H^{\Pi}_j}{1-\gamma}+\rho_j^2\sigma_j^2\frac{(H^{\Pi}_j)^2}{(1-\gamma)^2}\right]\\
&-\sum_{j=1}^{2}\frac{1}{2}(1-\gamma)\widetilde{\phi}_j^Vv_j\left[(\beta^V_j)^2+2\beta^V_j\sqrt{1-\rho_j^2}\sigma_j\frac{H^{\Pi}_j}{1-\gamma}+(1-\rho_j^2)\sigma_j^2\frac{(H^{\Pi}_j)^2}{(1-\gamma)^2}\right].
\end{aligned}
\end{equation}
Comparing the terms with and without multiplier $v_j$, one can certainly obtain the following  Riccati equations
\begin{equation}\label{e37}
\left\{
\begin{aligned}
&(H^{\Pi}_j)'=a_jH^{\Pi}_j+b_j(H^{\Pi}_j)^2+c_j,\quad && H^{\Pi}_j(0)=0,\ j=1, 2, \\
&(h^{\Pi})'=\kappa_1\theta_1H^{\Pi}_1+\kappa_2\theta_2H^{\Pi}_2+(1-\gamma)r,\quad && h^{\Pi}(0)=0,
\end{aligned}
\right.
\end{equation}
where 
\begin{eqnarray*}
\begin{aligned}
a_j=&-\kappa_j+\sigma_j(1-\gamma)(\beta^S_j\rho+\beta^V_j\sqrt{1-\rho^2})-\widetilde{\phi}^S_j\beta^S_j\rho\sigma_j-\widetilde{\phi}^V_j\beta^V_j\sqrt{1-\rho^2}\sigma_j,\\
b_j=&\frac{\sigma_j^2}{2}-\frac{\widetilde{\phi}^S_j\rho_j^2\sigma_j^2}{2(1-\gamma)}-\frac{\widetilde{\phi}^V_j(1-\rho_j^2)\sigma_j^2}{2(1-\gamma)},\\
c_j=&(1-\gamma)(\beta^S_j\lambda_1+\beta^V_j\lambda_2)-\frac{1}{2}\gamma(1-\gamma)((\beta^S_j)^2+(\beta^V_j)^2)-
\frac{(1-\gamma)(\beta^S_j)^2\widetilde{\phi}^S_j}{2}-\frac{(1-\gamma)(\beta^V_j)^2\widetilde{\phi}^V_j}{2}.
\end{aligned}
\end{eqnarray*}
\qed
\end{proof}
\section{Proof of Proposition \ref{prop4}}\label{A5}
\begin{proof}
Recall { that the utility loss}
$L^\Pi=1-\exp\left\{\frac{1}{1-\gamma}\left[(H_1^\Pi-H_1)v_1+(H_2^\Pi-H_2)v_2+(h^\Pi-h)\right]\right\}$.
If the investor chooses strategy $\Pi_1$, then $H_2^{\Pi_1}=H_2$ since the strategy does not affect on the uncertainty about the second component of the ambiguity. Assume that $\beta^S_1=p^S_1+p^S_2\widehat{H}_1$ and $\beta^V_1=p^V_1+p^V_2\widehat{H}_1$, where $p^S_j$ and $p^V_j$ $(j=1,2)$ are constants. Let $\widehat{H}_1$ be the function satisfying the equation $(\widehat{H}_1)'=\widehat{a}\widehat{H}_1+\widehat{b}\left(\widehat{H}_1\right)^2+\widehat{c}$, where $\widehat{a},\widehat{b}$ and $\widehat{c}$ are all constants,  we obtain the following system of equations in terms of $H^{\Pi}_1$ and $\widehat{H}_1$
\begin{equation}\label{e38}
\left\{
\begin{aligned}
&(H^{\Pi}_1)'=P_1+P_2\widehat{H}_1+P_3(H_1^{\Pi})^2+P_4H_1^{\Pi}+P_5\widehat{H}_1H_1^{\Pi}+P_6(H_1^{\Pi})^2, \quad &&H_1^{\Pi}(0)=0,\\
&(\widehat{H}_1)'=\widehat{a}\widehat{H}_1+\widehat{b}\widehat{H}_1+\widehat{c}, \quad &&\widehat{H}_1(0)=0,
\end{aligned}\right.
\end{equation}
where
\begin{equation}
\left\{
\begin{aligned}
P_1&=(1-\gamma)(\lambda_1p_1^S+\lambda_2p^V_2-\frac{1}{2}(\gamma+\widehat{\phi}_1^S))(p_1^S)^2-\frac{1}{2}(\gamma+\widehat{\phi}_1^V))(p_1^V)^2,\\
P_2&=(1-\gamma)(\lambda_1p_1^S+\lambda_2p^V_2-p_1^Sp_2^S(\gamma+\widehat{\phi}_1^S)-p_1^Vp_2^V(\gamma+\widehat{\phi}_1^V)),\\
P_3&=-\frac{1}{2}(1-\gamma)((\gamma+\widehat{\phi}_1^S)(p_2^S)^2+(\gamma+\widehat{\phi}_1^V)(p_2^V)^2),\\
P_4&=-\kappa_1+\sigma_1(\rho_1(1-\gamma-\widetilde{\phi}_1^S)p_1^S+\sqrt{1-\rho^2_1}(1-\gamma-\widetilde{\phi}_1^V)p_1^V),\\
P_5&=\sigma_1(\rho_1(1-\gamma-\widetilde{\phi}_1^S)p_2^S+\sqrt{1-\rho^2_1}(1-\gamma-\widetilde{\phi}_1^V)p_2^V),\\
P_6&=\frac{1}{2(1-\gamma)}\sigma_1^2(1-\gamma-\rho_1^2\widetilde{\phi}_1^S-(1-\rho_1^2)\widetilde{\phi}_1^V).
\end{aligned}
\right.
\end{equation}
It follows from \eqref{e24} and \eqref{e38} that
\begin{equation}\label{ss}
(H^{\Pi_1}_1-H_1)'+E(t)(H^{\Pi_1}_1-H_1)'=F(t),
\end{equation}
where
$E=-P_4-P_6(H^{\Pi_1}_1+H_1)-P_5\widehat{H}_1$
and
$F=P_1-c+(P_4-a)H_1+(K_6-b)H_1^2+P_2\widehat{H}_1+P_3\widehat{H}_1^2+P_5\widehat{H}_1H_1.$

{Solving} \eqref{ss}, we have
$$H^{\Pi_1}_1(t)-H_1(t)=e^{-\int_{t}^{T}E(s)ds}\int_{t}^{T}e^{\int_{s}^{T}E(\tau)d\tau}F(s)ds.$$
{ Under} the assumption that the suboptimal strategies are admissible,  the integral $\displaystyle e^{-\int_{t}^{T}E(s)ds}$ is bounded and positive. Furthermore,   the choice of strategy $\Pi_1$ means that  $\widetilde{\phi}_1^S = 0 $ and $\widetilde{\phi}_1^V = 0$. Thus,
\begin{equation*}
\begin{aligned}
F=&\frac{\gamma-1}{2\gamma^2(\phi^S_1+\gamma)}
\left(\lambda_1\phi^S_1+(\gamma+\phi^S_1)\rho_1\sigma_1(\widehat{H}_1-\frac{\gamma(\gamma+\phi_1^S-1)}{(\gamma-1)(\gamma+\phi_1^S)}H_1)\right)^2\\
&+\frac{\gamma-1}{2\gamma^2(\phi^V_1+\gamma)}
\left(\lambda_1\phi^V_1+(\gamma+\phi^V_1)\sqrt{1-\rho_1^2}\sigma_1(\widehat{H}_1-\frac{\gamma(\gamma+\phi_1^V-1)}{(\gamma-1)(\gamma+\phi_1^V)}H_1)\right)^2>0.
\end{aligned}
\end{equation*}
This leads to $H_1^\Pi-H_1>0$.
Following \eqref{e24}, we obtain $\displaystyle h^{\Pi_1}-h= \kappa_1\theta_1\int_0^t \left(H^{\Pi_1}_1(s)-H_1(s)\right)ds>0.$
Therefore, we have 
\begin{equation}
L^{\Pi_1}=1-\exp\left\{\frac{1}{1-\gamma}\left[(H_1^\Pi-H_1)v_1+(h^\Pi-h)\right]\right\}>0.
\end{equation}
Due to symmetry $L^{\Pi_2}>0$ can be proved in a similar way.

If $L^{\Pi_3}$ is chosen, from Propositions \ref{prop1} and \ref{prop2}, we { have}
$$H^{\Pi_3}_j(t)-H_j(t)=e^{-\int_{t}^{T}E_j(s)ds}\int_{t}^{T}e^{\int_{s}^{T}E_j(\tau)d\tau}F_j(s)ds, \ j=1, 2,$$
where
$F_j=\frac{\gamma-1}{2(\gamma+\phi^V_j)}\left(\lambda_j+\frac{\sigma_jH^{\Pi_3}_j\sqrt{1-\rho_j^2}(\gamma+\phi^V_j-1)}{\gamma-1}\right).$
It is straightforward to see that $F_j>0$, thus $L^{\Pi_3}>0$.
\qed
\end{proof}
\section{HJB in Section 4}\label{sec4}
In complete markets, the value function \eqref{max3} satisfies the robust HJB PDE
\begin{equation}\label{HJBC}
\begin{aligned}
&\sup_{\beta^S_i,\beta^V_i}\inf_{e^S_i,e^V_i}
\bigg\{
J_t+x\big(r+\sum_{j=1}^{2}(\beta^{S}_j\lambda_jv_j-\beta^{S}_ju_je^S_j
+\beta^{V}_j\mu_jv_j-\beta^{V}_ju_je^V_j)\big)J_x
+\sum_{j=1}^{2}\big(\kappa_j(\theta_j-v_j)-\rho_j\sigma_ju_je^S_j\\
&-\sqrt{1-\rho_j^2}\sigma_ju_je^V_j\big)J_{v_j}-\sum_{j=1}^{2}\left(\mu_j^U-\rho_j\psi^U_je^S_j-\sqrt{1-\rho_j^2}\psi^U_je^V_j\right)J_{u_j}
+\big(\mu^Y-\rho_2\psi_2^Ue^S_2u_1-\rho_1\psi_1^Ue^S_1u_2\\
&-\sqrt{1-\rho_2^2}\psi_2^Ue^V_2u_1-\sqrt{1-\rho_1^2}\psi_1^Ue^V_1u_2\big)J_y+\frac{1}{2}x^2\sum_{j=1}^{2}[(\beta^{S}_j)^2+(\beta^{V}_j)^2]v_jJ_{xx}
+\rho \rho_1\psi^U_1\beta^S_2u_2xJ_{xu_1}
\\
&+\sum_{j=1}^{2}\sigma_jv_jx\big(\beta^S_j\rho_j+\beta^V_j\sqrt{1-\rho_j^2}\big)J_{x v_j}+\rho \rho_2\psi^U_2\beta^S_1u_1xJ_{xu_2}+x\bigg(
\beta^S_1u_1(\rho\rho_2u_1\psi^U_2+\rho_1u_2\psi^U_1)\\
&+\beta^S_2u_2(\rho_2\psi_2^Uu_1+\rho\rho_1\psi^U_1u_2)+\sqrt{1-\rho_1^2}\beta^V_1\psi^U_1v_1+\sqrt{1-\rho_2^2}\beta^V_2\psi^U_2v_2
\bigg)J_{xy}
+\frac{1}{2}\sum_{j=1}^{2}\mathbf{\Sigma}_{jj}J_{v_jv_j}+ \mathbf{\Sigma}_{12}J_{v_1v_2}\\
&+\sum_{j=1}^{2}\mathbf{\Sigma}_{j,j+2}J_{v_ju_j}+\mathbf{\Sigma_{14}}J_{v_1u_2}+\mathbf{\Sigma}_{23}J_{v_2u_1}+\mathbf{\Sigma}_{15}J_{v_1y}
+\mathbf{\Sigma}_{25}J_{v_2y}+\frac{1}{2}\sum_{j=3}^{4}\mathbf{\Sigma}_{jj}J_{u_ju_j}+\mathbf{\Sigma}_{34}J_{u_1u_2}+\mathbf{\Sigma}_{35}J_{u_1y}\\
&+\mathbf{\Sigma}_{45}J_{u_2y}+\mathbf{\Sigma}_{55}J_{yy}+\sum_{j=1}^{2}\frac{(e^S_j)^2}{2\Psi^S_j}+\frac{(e^V_j)^2}{2\Psi^V_j}\bigg\}=0,
\end{aligned}
\end{equation}
where the symmetrical matrix $\mathbf{\Sigma}(t)$ is given as follows:
\begin{equation*}\label{matrix}
\begin{bmatrix}
\frac{1}{2}\sigma_1^2v_1 & \rho\rho_1\rho_2\sigma_1\sigma_2&
\sigma_1u_1\psi^U_1& \rho\rho_1\sqrt{1-\rho_2^2}\sigma_1\psi^U_2u_1& \sigma_1\rho_1u_1(\rho_1u_2\psi^U_1+\rho\rho_2u_1\psi^U_2)+(1-\rho_1^2)\sigma_1\psi_1^U\\
 & \frac{1}{2}\sigma_2^2v_2 & \rho\rho_1\rho_2\sigma_2\psi^U_1u_2 &
  \sigma_2u_2\psi_2^U &  \sigma_2\rho_2u_2(\rho_2u_1\psi^U_2+\rho\rho_1u_2\psi^U_1)+(1-\rho_2^2)\sigma_2\psi_2^U\\
 &  & \frac{1}{2}(\psi^U_1)^2 & \rho\rho_1\rho_2\psi^U_1\psi^U_2 & (\psi^U_1)^2u_2+\rho\rho_1\rho_2\psi_1^U\psi_2^Uu_1\\
 &  & &\frac{1}{2}(\psi^U_2)^2& (\psi^U_2)^2u_1+\rho\rho_1\rho_2\psi_1^U\psi_2^Uu_2\\
 & & & & \frac{1}{2}\left(u_1^2(\psi^U_2)^2+u_2^2(\psi^U_1)^2+\rho\rho_1\rho_2\psi^U_1\psi^U_2\right)
\end{bmatrix}.
\end{equation*}

\section{Detection-Error Probabilities}\label{A6}
Define the conditional characteristic functions
$$f_1(\omega,t,T)=\mathrm{E}^{\mathbb{P}}[\exp(i\omega \xi_1(T))|\mathcal{F}_t^{S,V_1,V_2}]=\mathrm{E}^{\mathbb{P}}[(\mathcal{Z}_1(T))^{i\omega}|\mathcal{F}_t^{S,V_1,V_2}]$$
and
$$f_2(\omega,t,T)=\mathrm{E}^{\mathbb{P}^e}[\exp(i\omega \xi_1(T))|\mathcal{F}_t^{S,V_1,V_2}]=\mathrm{E}^{\mathbb{P}^e}[(\mathcal{Z}_1(T))^{i\omega}|\mathcal{F}_t^{S,V_1,V_2}]
=\mathrm{E}^{\mathbb{P}}[(\mathcal{Z}_1(T))^{i\omega+1}|\mathcal{F}_t^{S,V_1,V_2}],
$$
where $i^2=-1$ and $\omega$ is the transform variable. Denote $\mathbf{e_t}=[e^S_1,e^S_2,e_1^V,e^V_2]$
and
$$\mathbf{\sigma}=\left[\sigma_1\sqrt{V_1}\rho_1,\sigma_2\sqrt{V_2}\rho_2,\sigma_1\sqrt{V_1}\sqrt{1-\rho_1^2},\sigma_2\sqrt{V_2}\sqrt{1-\rho_2^2}\right].$$
Since the conditional characteristic functions are martingales, the Feynman-Kac theorem implies that $f_1$ and $f_2$
satisfy the same PDE
\begin{equation}\label{PDEF}
\frac{\partial f}{\partial t}+\sum_{j=1}^{2}\kappa_j(\theta_j-V_j)\frac{\partial f}{\partial V_j}+\frac{1}{2}\mathcal{Z}_1(t)^2||\mathbf{e_t}||^2
\frac{\partial^2 f}{\partial \mathcal{Z}_1^2}+\frac{1}{2}\sum_{j=1}^{2}\sigma^2_jV_j \frac{\partial^2 f}{\partial V_j^2}
-\sum_{j=1}^{2}\mathcal{Z}_1(t)\mathbf{\sigma}{\mathbf{e_t}}{\hskip -3pt}^{T}\frac{\partial^2 f}{\partial {\mathcal{Z}_1}\partial V_j}=0
\end{equation}
with different terminal conditions $f_1(\omega,T,T)=\mathcal{Z}_1^{i\omega}(T)$ and $f_2(\omega,T,T)=\mathcal{Z}_1^{i\omega+1}(T)$, respectively. Conjecturing the solution in the form $f=\mathcal{Z}_1^{i\omega}(t)\exp(C_1(t)V_1+C_2(t)V_2+D(t))$ and substituting it into \eqref{PDEF}, we obtain
$$
C^{'}_{j}(t)-\kappa_jC_j(t)+\frac{1}{2}i\omega(i\omega-1)(q^S_j+q^V_j)^2+\frac{1}{2}\sigma_j^2C_j^2(t)-i\omega C_j(t)\sigma_i(q^S_j\rho_j+q^V_j\sqrt{1-\rho_j^2})=0,\  C_j(T)=0,\  j=1, 2,
$$
$$
D^{'}(t)+\kappa_1\theta_1C_1(t)+\kappa_2\theta_2C_2(t)=0,\ D(T)=0.
$$
Similarly, we can obtain the solution for $f_2$ by replicating $i\omega$ with $i\omega+1$ in the above equations. Thus, we have the detection-error probability
$$\varepsilon_T(\phi^S_1,\phi^S_2,\phi^V_1,\phi^V_2)=\frac{1}{2}-\frac{1}{2\pi}\int_{0}^{\infty}\left(\mathrm{Re}\left[\frac{f_1(\omega,0,T)}{i\omega}\right]-\mathrm{Re}\left[\frac{f_2(\omega,0,T)}{i\omega}\right]\right)d\omega.$$

\end{section}
\end{appendices}
}

{\normalsize
\doublespacing

}

\end{document}